\documentclass[aps,prx,superscriptaddress,twocolumn,longbibliography]{revtex4-1}
\usepackage{amstext}
\usepackage{amsthm}
\usepackage{amssymb}
\usepackage{amsmath}
\usepackage{color}
\usepackage{xcolor}

\usepackage[colorlinks=true, urlcolor=blue, citecolor=blue,linkcolor=blue,citebordercolor={1 0 0},linkbordercolor={0 0 1}]{hyperref}

\usepackage{graphicx}
\graphicspath{{./images/},{./imagesappendix/}}

\makeatletter
\theoremstyle{plain}
\newtheorem{thm}{\protect\theoremname}
\theoremstyle{remark}
\newtheorem{claim}[thm]{\protect\claimname}

\ifx\proof\undefined
\newenvironment{proof}[1][\protect\proofname]{\par
	\normalfont\topsep6\p@\@plus6\p@\relax
	\trivlist
	\itemindent\parindent
	\item[\hskip\labelsep\scshape #1]\ignorespaces
}{%
	\endtrivlist\@endpefalse
}
\providecommand{\proofname}{Proof}
\fi

\theoremstyle{definition}
\newtheorem{defn}[thm]{\protect\definitionname}
\providecommand{\definitionname}{Definition}

\renewcommand{\eqref}[1]{Eq.~(\ref{#1})} % Reference to equation

\newcommand\numberthis{\addtocounter{equation}{1}\tag{\theequation}}
\newcommand{\bra}[1]{\langle #1|}
\newcommand{\ket}[1]{|#1 \rangle}
\makeatother
\newcommand{\braket}[2]{\langle #1 \vert #2 \rangle}

\definecolor{KB}{rgb}{0.4,0.3,0.9}

\definecolor{THc}{rgb}{0.9,0.3,0.2}

\newcommand{\revA}[1]{{#1}}
\makeatother

\usepackage{babel}
\providecommand{\claimname}{Claim}
\providecommand{\theoremname}{Theorem}

\newcommand{\idg}[1]{{\bfseries #1)}}
\newcommand{\subfigimg}[3][,]{%
	\setbox1=\hbox{\includegraphics[#1]{#3}}% Store image in box
	\leavevmode\rlap{\usebox1}% Print image
	\rlap{\hspace*{-3pt}\raisebox{\dimexpr\ht1-0.5\baselineskip}{{ \textsf{#2}}}}% Print label
	\phantom{\usebox1}% Insert appropriate spcing
}

\begin{document}
\title{Noisy intermediate-scale quantum algorithm for semidefinite programming}
\author{Kishor Bharti}
\email{kishor.bharti1@gmail.com}
\affiliation{Centre for Quantum Technologies, National University of Singapore, 3 Science Drive 2, Singapore 117543}
\author{Tobias Haug}
\affiliation{QOLS, Blackett Laboratory, Imperial College London SW7 2AZ, UK}
\author{Vlatko Vedral}
\affiliation{Centre for Quantum Technologies, National University of Singapore, 3 Science Drive 2, Singapore 117543} 
\affiliation{Clarendon Laboratory, University of Oxford, Parks Road, Oxford OX1 3PU, United Kingdom }

\author{Leong-Chuan Kwek}
\affiliation{Centre for Quantum Technologies, National University of Singapore, 3 Science Drive 2, Singapore 117543}
\affiliation{MajuLab, CNRS-UNS-NUS-NTU International Joint Research Unit, UMI 3654, Singapore}
\affiliation{National Institute of Education,
Nanyang Technological University, 1 Nanyang Walk, Singapore 637616}
\affiliation{School of Electrical and Electronic Engineering
Block S2.1, 50 Nanyang Avenue, 
Singapore 639798 }

\begin{abstract}
Semidefinite programs (SDPs) are convex optimization programs with vast applications in control theory, quantum information, combinatorial optimization and operational research. Noisy intermediate-scale quantum (NISQ) algorithms aim to make an efficient use of the current generation of quantum hardware. However, optimizing variational quantum algorithms is a challenge as it is an NP-hard problem that in general requires an exponential time to solve and can contain many far from optimal local minima. 
Here, we present a current term NISQ algorithm for solving SDPs. The classical optimization program of our NISQ solver is another SDP over a lower dimensional ansatz space. We harness the SDP based formulation of the Hamiltonian ground state problem to design a NISQ eigensolver. Unlike variational quantum eigensolvers, the classical optimization program of our eigensolver is convex, can be solved in polynomial time with the number of ansatz parameters and every local minimum is a global minimum.
\revA{We find numeric evidence that NISQ SDP can improve the estimation of ground state energies in a scalable manner.
Further, we efficiently solve constrained problems to calculate the excited states of Hamiltonians, find the lowest energy of symmetry constrained Hamiltonians and determine the optimal measurements for quantum state discrimination.}
We demonstrate the potential of our approach by finding the largest eigenvalue of up to $2^{1000}$ dimensional matrices and solving graph problems related to quantum contextuality. We also discuss NISQ algorithms for rank-constrained SDPs. Our work extends the application of NISQ computers onto one of the most successful algorithmic frameworks of the past few decades.
\end{abstract}
\maketitle

\section{Introduction}
%\noindent 

The panorama of quantum computing has been transformed enormously
in the last forty years. Once acknowledged as a theoretical
pursuit, quantum computers with a few dozen qubits are now a reality. Advancement at the hardware frontier has led to the demonstration of ``computational quantum supremacy'' for contrived tasks~\cite{arute2019quantum,zhong2020quantum}. We sit at the edge of the noisy intermediate-scale quantum (NISQ) era~\cite{preskill2018quantum,bharti2021noisy}.
In recent years, significant effort has been put towards designing algorithms for practically relevant tasks which can be implemented on NISQ devices~\cite{bharti2021noisy,cerezo2020variational}.
Canonical examples of these NISQ algorithms are variational quantum algorithms (VQAs) such as the variational quantum eigensolver (VQE)~\cite{peruzzo2014variational,mcclean2016theory,kandala2017hardware} and the quantum approximate optimization algorithm (QAOA)~\cite{farhi2014quantum,farhi2016quantum}. NISQ algorithms have been developed for various tasks such as finding the ground state of Hamiltonians~\cite{peruzzo2014variational,mcclean2016theory,kandala2017hardware,mcclean2017hybrid,kyriienko2020quantum,parrish2019quantum,bespalova2020hamiltonian,huggins2020non,takeshita2020increasing,stair2020multireference,motta2020determining,seki2020quantum,bharti2020quantum,bharti2020iterative,cervera2021meta}, combinatorial optimization~\cite{farhi2014quantum,farhi2016quantum}, quantum simulation~\cite{li2017efficient,yuan2019theory,benedetti2020hardware,bharti2020simulator,barison2021efficient,commeau2020variational,heya2019subspace,cirstoiu2020variational,gibbs2021long,lau2021quantum,haug2020generalized,otten2019noise,lim2021fast,lau2021nisq} quantum metrology~\cite{meyer2020variational,meyer2021fisher} and machine learning~\cite{schuld2019quantum,havlivcek2019supervised, kusumoto2019experimental,farhi2018classification, mitarai2018quantum}. These algorithms have been investigated in detail with a comprehensive exposition on possible hurdles~\cite{mcclean2018barren,sharma2020trainability,cerezo2020cost,wang2020noise,bittel2021training,huang2019near} and corresponding countermeasures~\cite{huang2019near,haug2021capacity,haug2021optimal,larocca2021diagnosing}.  A thorough study of possible applications
of NISQ devices is expected to unravel the potential as well as limitations of such devices. Moreover, in the quest for practical quantum advantage in the NISQ era it
is pertinent to investigate novel NISQ algorithms for practically relevant tasks.

A major challenge in the NISQ era is the optimization program for VQAs, where a classical optimizer is searching for the parameters of a quantum state that minimizes a cost function. For the VQA to be successful, one requires an ansatz that is expressible enough to approximate the optimal solution of the corresponding optimization problem.
However, even if such ansatz has been found, the optimization program of the VQA is NP-hard and in contrast to classical neural networks the optimization landscape contains numerous far from optimal local minima~\cite{bittel2021training,anschuetz2021critical,you2021exponentially}. The highly non-convex nature of landscape renders optimization difficult even for VQAs involving logarithmically many qubits or classically easy problems such as free fermions. 
\revA{Further, VQAs struggle to optimize problems where the solution space is constrained due to symmetries such as for chemistry problems~\cite{higgott2019variational,mcclean2016theory,rubin2018application,ryabinkin2018constrained,greene2021generalized,kuroiwa2021penalty}.}

In the last few decades, semidefinite programs (SDPs) have led to ground breaking developments in mathematical optimization~\cite{vandenberghe1996semidefinite,wolkowicz2012handbook}.
The study of SDPs has uncovered numerous applications in theoretical computer science,
control theory and operations research. Many problems in quantum
information such as state discrimination~\cite{skrzypczyk2019all,bae2015quantum,jevzek2002finding}, dimension witness~\cite{ray2021graph} and self-testing~\cite{bharti2019local,bharti2019robust,bharti2021graph,yang2014robust,bancal2015physical} can be investigated using SDPs. 
While SDPs can be solved efficiently in polynomial time on classical computers, high dimensional problems may still be out of scope for classical computers. For example, finding the ground state of a Hamiltonian can be framed as a SDP, however it is intractable for classical computers due to the exponential scaling of the dimension of the problem.
To explore possible quantum advantages for solving SDPs with quantum computers, quantum SDP solvers have been
studied comprehensively~\cite{brandao2017quantum,van2017quantum,van2018improvements,kerenidis2020quantum,brandao2017quantum2,chakrabarti2020quantum,Apeldoorn2018}. However, existing quantum SDP solvers cannot be executed on NISQ devices as they require extensive quantum resources.

Here, we propose the NISQ SDP Solver (NSS) as a hybrid quantum-classical algorithm to solve SDPs. The NSS encodes a SDP onto a quantum computer combined with an optimization routine on a classical computer.
The classical optimization part of the NSS is also a SDP with its dimension given by the size of the ansatz space.  The quantum computational part of the NSS has no classical-quantum feedback loop and requires the quantum computer only for estimating overlaps, which can be done efficiently on current NISQ devices.
We showcase the NSS for a wide range of problems. We design a NISQ SDP based quantum eigensolver (NSE) to find the ground state of quantum Hamiltonians. In contrast to VQE, the classical optimization part of the NSE can be solved in polynomial time without the local minima problem.
\revA{We find numerical evidence that the NSE improves the estimation of the ground state energy by a constant factor for any number of qubits for a non-integrable one-dimensional Ising model combined with a quantum annealing ansatz.
Furthermore, the NSS is capable of efficiently implementing constraints in the optimization program in order to calculate excited states and solve symmetry constrained problems. In addition, we provide an NSS for determining optimal measurements for quantum state discrimination.}
The NSS can be also find the largest eigenvalue of matrices, which we demonstrate for matrices as large as $2^{1000}$.
Finally, we show the NSS for various important problems related to quantum information such as Bell non-local games and the Lovász Theta number. We also provide the extension of the NSS for rank-constrained SDPs.

\revA{One can argue that SDPs are solvable efficiently in polynomial time on classical computers and hence why should one construct a NISQ algorithm for SDPs. We would like to stress that the polynomial runtime is in terms of the input matrix size and the number of constraints. For problems with exponential input size, polynomial of exponential would be still exponential and hence classical SDPs would be unable to process such cases. To begin, the Hamiltonian ground state problem is an SDP, but that does not mean it is tractable. In other words, 
while SDPs can be solved in polynomial time and memory, when the problem scales exponentially classical computers are unable to process it. For example,  no classical computer in the world is able to store a problem of size $2^{60}$. However, quantum computers can store this vector efficiently within a quantum state of 60 qubits, which is the current state of the art of NISQ computers.
Here, our NSS algorithm offers the potential to outperform any classical SDP solver.}

% Encoding problems into quantum computers efficiently is a well known open problem for most (fault-tolerant) quantum algorithms.
% When no efficient encoding into quantum computer is known, we suffer the same input encoding problem as all other quantum algorithms.

% However, given that the encoding is found, our algorithm can be immediately applied. We demonstrate this for quantum native problems that have a natural encoding into quantum computers, which are the ground state problem or our newly added unambigious state discrimination task. Here, our algorithm can immediately put to use.}

\section{Background}
We now highlight the key difference between VQA and NSS in Fig.\ref{fig:sketch} by using the ground state problem as an example.
Finding the ground state of a Hamiltonian can be framed as a SDP using density matrices (see program \ref{eq:Ham_SDP_primal_1}), however the SDP suffers from exponential scaling of the dimension of the quantum state and thus it is difficult to solve on classical computers.
To address the scaling of the quantum state, VQE and NSE map the quantum state onto a quantum computer. 
VQE uses a quantum circuit parameterized by the parameter $\boldsymbol{\theta}$. Then, the VQE minimizes the energy of the quantum state by variationally adjusting the parameter $\boldsymbol{\theta}$ via a classical optimization routine in a feedback loop. However, this minimization task is challenging as the corresponding optimization program is non-convex and NP-hard~\cite{bittel2021training}. %In particular, the energy as function of $\boldsymbol{\theta}$ is non-convex with persistent far from optimal local minima~\cite{bittel2021training}, where optimization routines struggle to find the global minima.
\revA{In contrast to classical neural networks, the optimization landscape of VQEs is characterised by local minimas far from the global minima, where optimization routines are unlikely to find reasonable approximations of the global minima~\cite{anschuetz2021critical}. Only for overparameterized circuits the landscape becomes favorable, which for most types of circuits requires an exponential amount of parameters and exponentially deep circuits~\cite{larocca2021theory,haug2021capacity}. }

In contrast, the NSS uses a hybrid density matrix~\eqref{eq:Ansatz_1}. It consists of a linear combination of a $M$-dimensional set of ansatz quantum states with classical combination coefficients $\boldsymbol{\beta}$. In the NSS, the $\boldsymbol{\beta}$ parameters are optimised to minimize the energy. 
The optimization of the coefficients $\boldsymbol{\beta}$ is another SDP with dimension $M$ only, which can be efficiently optimised in polynomial time on a classical computer. The ansatz preserves the convexity of the landscape and hence any local minimum is also a global minimum.

\begin{figure}[htbp]
	\centering	
	\subfigimg[width=0.45\textwidth]{}{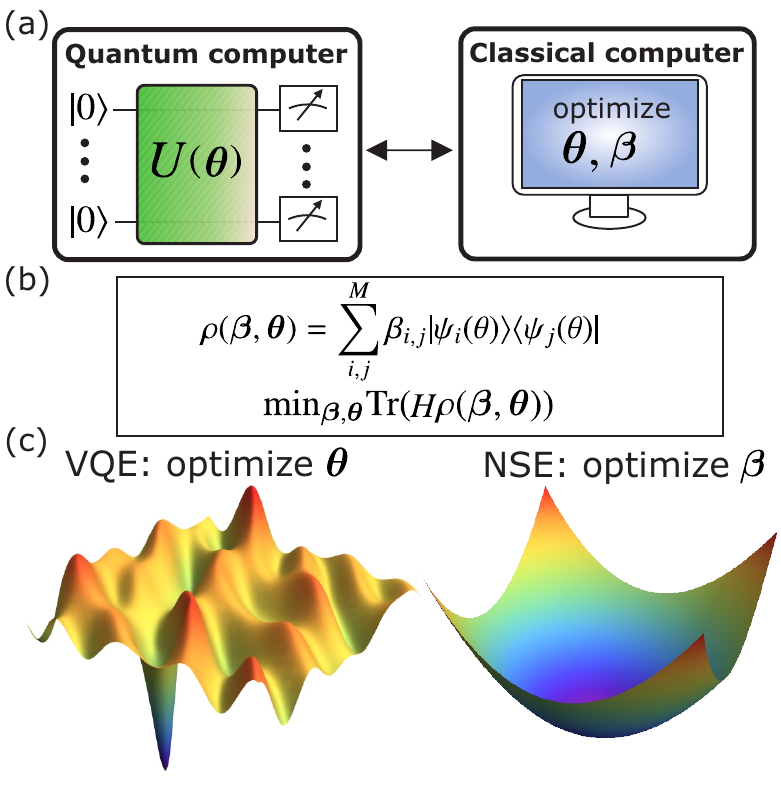}
	\caption{\idg{a} Hybrid quantum-classical computing approach to NISQ. Quantum computer prepares and measures $M$ quantum states $\vert \psi_i(\boldsymbol{\theta})\rangle$. Classical computer is used to optimize hybrid density matrix $\rho(\boldsymbol{\beta},\boldsymbol{\theta})=\sum_{i,j}^M\beta_{i,j}\vert \psi_{i}(\boldsymbol{\theta})\rangle\langle\psi_{j}(\boldsymbol{\theta})\vert$, which depends on the parameters for the parameterized quantum circuit $\boldsymbol{\theta}$ and the combination coefficients $\boldsymbol{\beta}$. 
	\idg{b} Optimization task to find ground state of Hamiltonian $H$ by minimizing parameters $\boldsymbol{\beta},\boldsymbol{\theta}$ in respect to energy $\text{Tr}(H\rho(\boldsymbol{\beta},\boldsymbol{\theta}))$.
	\idg{c} Landscape of optimization task for $\boldsymbol{\beta}$ and $\boldsymbol{\theta}$. Optimization of circuit parameters $\boldsymbol{\theta}$ in the variational quantum eigensolver (VQE) is NP-hard with a non-convex landscape and persistent far from optimal local minima. Optimisation of coefficients $\boldsymbol{\beta}$ in our NISQ SDP based eigensolver (NSE) can be solved in polynomial time by a semidefinite program (SDP), where the optimisation landscape is convex. Due to the convexity of the optimisation landscape, any local minimum is also a global minimum.
	}
	\label{fig:sketch}
\end{figure}

\section{Semidefinite Program}
SDPs can be thought of as generalization of the ``standard form" linear
programming (LP). The ``standard form" of LP is given by
\begin{gather*}
 \text{min}\text{ }c^{T}x \numberthis \label{eq:LP1}\\
\text{s.t }a_{i}^{T}x=b_{i}\quad\forall i\in[m] \\
x \in \mathbb{R}^n_{+}\,,
\end{gather*}
where the set $[m]$ is given by $[m]\equiv\left\{ 1,2,3,\ldots,m\right\}$.  Here, $\mathbb{R}^n_{+} := \{x\in \mathbb{R}^n \vert x \geq 0\}$. The set $\mathbb{R}^n_{+}$ is known as nonnegative orthant. The phrase ``standard form" hints that there are other possible non-standard representations  of LPs. Any LP in non-standard form, however, can be converted into standard form by following a few tricks. These tricks include change of variables, transforming the inequalities into equalities and switching maximum to minimum.   

In a SDP, the non-negativity
constraint $x\ge 0$ is replaced by positive semidefinite cone constraint
$X\succcurlyeq0.$ SDPs involve optimization of a linear function
of matrix $X$ over the affine slice of the cone of  positive semidefinite
matrices. The standard form of a SDP is given by
\begin{gather*}
 \text{min}\text{ }Tr\left(CX\right) \numberthis \label{eq:SDP1}\\
\text{s.t. }Tr\left(A_{i}X\right)=b_{i}\quad\forall i\in[m]\\
X \in \mathcal{S}_+^{n} \,.
\end{gather*}
Here, $\mathcal{S}_+^{n}$ denotes the set of $n\times n$ symmetric positive semidefinite matrices. Mathematically speaking, $\mathcal{S}^n_{+} := \{X\in \mathcal{S}^n \vert X \succcurlyeq 0\}$. The matrices $C$ and $A_i$ belong to the set of symmetric matrices $\mathcal{S}^n$ for $i\in [m]$.  The $i$-th element of vector $b\in \mathbb{R}^m$ is denoted by $b_i$.

Duality is one of the oldest and most fruitful ideas in mathematics. The duality principle of mathematical optimization theory suggests
that mathematical optimization problems can be viewed from either
of two perspectives, namely the primal problem or the dual problem.
For a given primal minimization problem $P$, the solution to the
corresponding dual problem problem $D$ provides a lower bound to
the solution of $P$. The standard form of the dual of the SDP in
program~\ref{eq:SDP1} is given by
\begin{gather*}
 \text{max }b^{T}y \numberthis \label{eq:SDP_dual_1}\\
\text{s.t. }\sum_{i=1}^{m}y_{i}A_{i}\preccurlyeq C\\
y \in \mathbb{R}^m\,.
\end{gather*}
The SDPs in programs \ref{eq:SDP1} and \ref{eq:SDP_dual_1} constitute
a primal-dual pair. SDPs can be extended to complex-valued matrices via a cone of Hermitian positive semidefinite matrices i.e. $X \in \mathcal{H}_+^{n}$. Here, $\mathcal{H}_+^{n}$ denotes the set of $n\times n$ Hermitian positive semidefinite matrices. Since SDPs for complex-valued matrices are more general than SDPs for real valued matrices, we will consider the former case in this work.

\section{The NSS} \label{sec:NSS}
We now outline the NSS, which consists of three distinct steps, namely ansatz selection, overlap measurement and post-processing.
First, we select a set of quantum states $\mathbb{S}=\left\{ \vert\psi_{j}\rangle\in\mathcal{H}\right\} _{j}$ over a Hilbert space $\mathcal{H}$, where the set contains $M$ quantum states $\mathbb{\left|S\right|} = M$. Now, our NISQ semidefinite programming solver (NSS) uses the following
hybrid density matrix ansatz
\begin{equation}
X_\beta=\sum_{\left(\vert\psi_{i}\rangle,\vert\psi_{j}\rangle\right)\in\mathbb{S}\times\mathbb{S}}\beta_{i,j}\vert\psi_{i}\rangle\langle\psi_{j}\vert,\label{eq:Ansatz_1}
\end{equation}
where $\beta_{i,j}\in\mathbb{C}$. Note that the quantum states in $\mathbb{S}$ are prepared by a quantum system while the coefficients
$\beta_{i,j}$ are stored on some classical device as matrix $\beta$.
For $\beta \in \mathcal{H}_+^{M}$, we have $X_{\beta} \in \mathcal{H}_+^{n}$ (see Appendix \ref{sec:proof_density}).  We now assume that $C$ and constraint matrices $A_i$ of program~\ref{eq:SDP1} can be written as a sum of unitaries 
\begin{gather*}
C=\sum_{k}s_{k}U_{k}\\
A_{i}=\sum_{l}f_{i,l}U_{l}^{(i)}\,.
\end{gather*}
As second step of the NSS, we measure the following overlaps on the quantum system
\begin{gather}
\mathcal{D}_{a,b}=\sum_{k}s_{k}\langle\psi_{b}\vert U_{k}\vert\psi_{a}\rangle\label{eq:D_matrix}\\
\mathcal{E}_{a,b}^{(i)}=\sum_{l}f_{i,l}\langle\psi_{b}\vert U_{l}^{(i)}\vert\psi_{a}\rangle \,.\label{eq:E_matrix}
\end{gather}
The overlaps can be measured using the Hadamard test or with direct measurement
methods~\cite{mitarai2019methodology}. 
An alternative NISQ-friendly method that requires only sampling in the computational basis has been proposed in~\cite{bharti2020iterative}, which we use in the following. This method assumes that the unitaries $U_k$ and $U_l^{(i)}$ are Pauli strings $P=\bigotimes_{j=1}^N\boldsymbol{\sigma}_j$ with $\boldsymbol{\sigma}\in \{I,\sigma^x,\sigma^y,\sigma^z\}$. As the Pauli strings form a complete basis, any matrix can be decomposed into a linear combination of Pauli strings.
Further, we assume that the ansatz space is generated by an initial state $\ket{\psi}$ and a set of $M$ different Pauli strings $\{P_1,\dots,P_M\}$ via $\mathbb{S}=\{P_j\ket{\psi}\}_{j=1}^M$.
Now, each overlap element in \eqref{eq:D_matrix}, \eqref{eq:E_matrix} can be written as a sum of expectation values of Pauli strings $\bra{\psi}P_1P_2P_3\ket{\psi}=a\bra{\psi}P'\ket{\psi}$, where we use that a product of Pauli strings can be written as a single Pauli string $P'$ with a prefactor $a \in \{+1,-1,+i,-i\}$. 
Then, one can efficiently calculate the overlap elements by measuring the expectation values of Pauli strings. On NISQ computers, this can be done by preparing the initial state $\vert \psi \rangle $, performing single-qubit rotations into the eigenbasis of the Pauli operator and then sampling in the computational basis $N_\text{Q}$ times. 
%We emphasize that for this ansatz one only needs to prepare the initial state .
%We emphasize that for the choice of quantum states of the form $\mathbb{S}=\{P_j\ket{\psi}\}_{j=1}^M$, one only needs to prepare the reference state $\vert \psi \rangle $. Other states are related to the reference state via single layer Pauli unitaries. In such cases, the second step of our algorithm requires sampling $\vert \psi \rangle $ in a set of Pauli rotated basis elements.
\revA{We note that sampling from quantum circuits is general intractable for classical computers~\cite{aaronson2016complexity}. The additive error $\Delta P$ of estimating the expectation value of Pauli strings scales as $\Delta P\propto N_\text{Q}^{-\frac{1}{2}}$ according to Hoeffding's inequality and is independent of the number of qubits~\cite{huang2019near}. }

The third and final step of the NSS consists of post-processing on a classical computer. Here, we write the standard form primal SDP in terms of the measured overlaps
\begin{gather*}
\min\text{ }Tr\left(\mathcal{D}\beta\right)\numberthis\label{eq:Ansatz_SDP_primal}\\
\text{s.t. }Tr\left(\mathcal{E}^{(i)}\beta\right)= b_{i}\quad\forall i\in\left[m\right]\\
\beta \in \mathcal{H}_+^{M}\,.
\end{gather*}
This is a SDP over $\beta$ with the corresponding hybrid density matrix is given by \eqref{eq:Ansatz_1}. The dual of the SDP in program \ref{eq:Ansatz_SDP_primal}
is given by
\begin{gather*}
\text{max }b^{T}y \numberthis\label{eq:Ansatz_SDP_Dual}\\
\text{s.t. }\sum_{i=1}^{m}y_{i}\mathcal{E}^{(i)}\preccurlyeq\mathcal{D}\,.
\end{gather*}
The SDPs in program \ref{eq:Ansatz_SDP_primal} and \ref{eq:Ansatz_SDP_Dual}
constitute a primal-dual pair over the ansatz space that can be solved in polynomial time on a classical computer. In the case where the ansatz states $\vert \psi_i \rangle $ are linear independent and cover the whole Hilbert space, the programs~\ref{eq:Ansatz_SDP_primal} and~\ref{eq:Ansatz_SDP_Dual} over the ansatz space recover the SDPs over the entire space corresponding to programs~\ref{eq:SDP1} and~\ref{eq:SDP_dual_1}.

\section{Extensions to Rank-constrained SDPs}
Rank constrained SDPs find numerous applications in combinatorics~\cite{marianna2013complexity}, control theory~\cite{boyd1994linear} and quantum information~\cite{ray2021graph}. Many problems in optimization theory can be modelled as rank-constrained SDPs~\cite{vandenberghe1996semidefinite,anjos2011handbook}. The rank-constraint turns the optimization program into a non-convex and NP-hard problem, which is in general intractable to solve. The optimization program for rank constrained SDPs is given by
\begin{gather*}
 \text{min}\text{ }Tr\left(CX\right) \numberthis \label{eq:SDP_rank_k}\\
\text{s.t. }Tr\left(A_{i}X\right)=b_{i}\quad\forall i\in[m]\\
\text{rank}(X) \leq k \\
X \in \mathcal{S}_+^{n} \,.
\end{gather*}
Notice that the only difference between the regular SDP in program~\ref{eq:SDP1} and the rank-constrained SDP in program~\ref{eq:SDP_rank_k} is the presence of the rank constraint $rank(X) \leq k$ in the latter. Intuitively speaking, the aforementioned rank constraint means that the optimizer of the program~\ref{eq:SDP_rank_k} must have a rank of at most $k$. The famous Max-Cut problem can be modelled as a rank constrained SDP for $k=1$ (see Appendix~\ref{sec: rank_constrained}). 
We defer the extension of the NSS for the rank constrained SDPs in Appendix~\ref{sec: rank_constrained}. 

\section{Examples}
We now demonstrate the NSS for various problems of interest.

\subsection{Ground State Problem}

First, we demonstrate the NSE solver to find the ground state of Hamiltonians.
For Hamiltonian $H=\sum_k s_kU_k$ and density matrix $\rho$, the
problem of finding the ground state can be written
as
\begin{gather*}
\min\text{ }Tr\left(\rho H\right)\label{eq:Ham_SDP_primal_1}\numberthis\\
\text{s.t. }Tr\left(\rho\right)=1\\
\rho\succcurlyeq0\,.
\end{gather*}
The dual formulation for Program \ref{eq:Ham_SDP_primal_1}
is given by
\begin{gather*}
\max\text{ }\lambda\numberthis\label{eq:Ham_SDP_dual_1}\\
\text{s.t. }\left(H-\lambda I\right)\succcurlyeq0\,.
\end{gather*}
The optimum value of $\lambda$ corresponds to the ground state energy.
In the ansatz space generated by $\mathbb{S}$, the
primal SDP for the Hamiltonian ground state problem is given by
\begin{gather*}
\min\text{ }Tr\left(\beta\mathcal{D}\right)\numberthis\label{eq:Ham_SDP_Ansatz_Primal_1}\\
\text{s.t. }Tr\left(\beta\mathcal{E}\right)=1\\
\beta \in \mathcal{H}_+^{M}\,,
\end{gather*}
with $\mathcal{E}_{a,b}=\braket{\psi_b}{\psi_a}$ and $\mathcal{D}_{a,b}=\sum_k s_k\bra{\psi_b}U_k\ket{\psi_a}$.
The dual program of Program~\ref{eq:Ham_SDP_Ansatz_Primal_1} is given by
\begin{gather*}
\max\text{ }\lambda\numberthis\label{eq:Ham_SDP_Ansatz_dual_1}\\
\text{s.t. }\left(\mathcal{D}-\lambda\mathcal{E}\right)\succcurlyeq0\,.
\end{gather*}
Notice that the primal optimization program over $\beta$ is convex
and hence exhibits a unique minimum value. This is unlike the case of VQE,
where optimization is non-convex and there can be multiple local minima~\cite{bittel2021training}. In the NSE, the classical optimization finds the optimal solution in time polynomial in the number of ansatz parameters, given that the solution is contained in the ansatz space. On the other hand, in VQAs such as VQE and QAOA, even if the optimal solution is contained in the ansatz space, the classical optimization can be NP-hard~\cite{bittel2021training}.
\revA{ Moreover, the primal optimal solution is equal to the dual optimal solution.
\begin{claim}
The primal optimal solution corresponding to the SDP in Program \ref{eq:Ham_SDP_Ansatz_Primal_1} is equal to its dual optimal solution (Program \ref{eq:Ham_SDP_Ansatz_dual_1}). In other words, the SDP in \eqref{eq:Ham_SDP_Ansatz_Primal_1} admits strong duality. 
\end{claim} 
\begin{proof}
To show that the SDP in Program \ref{eq:Ham_SDP_Ansatz_Primal_1} admits strong duality, we need to show that the interior is non-empty~\cite{boyd2004convex}. The non-empty interior can be established via the existence of a full rank feasible solution for Program~\ref{eq:Ham_SDP_Ansatz_Primal_1}. Consider 
$$\mathcal{L} = \mathbb{I}_M.$$
Clearly, $\mathcal{L}$ is feasible and full rank. This completes the proof.
\end{proof}
}

\revA{Now we demonstrate the NSE for finding the ground state of a non-integrable model, namely the one-dimensional Ising model of $N$ qubits with transverse field $h$ and longitudinal field $g$
\begin{equation}\label{eq:Ising}
    H_\text{ising}=H_z+H_x=-\sum_{n=1}^N[ \sigma^z_{n}\sigma^z_{n+1}+g\sigma^z_n]-\sum_{n=1}^N h\sigma^x_n\,
\end{equation}
with $H_z=-\sum_{n=1}^N[ \sigma^z_{n}\sigma^z_{n+1}+g\sigma^z_n]$ and $H_x=-\sum_{n=1}^N h\sigma^x_n$.
First, we prepare an initial state on the quantum computer, which we then use to construct the ansatz space. As initial state, we either choose a hardware efficient circuit $\ket{\psi_\text{rand}}$ consisting $p$ layers of randomized $y$-rotations and CNOT gates arranged in a chain topology (see Appendix~\ref{sec:hardware_efficient}) %and a QAOA ansatz $\ket{\psi_\text{QAOA}(\gamma,\delta)}=\prod_{k=1}^p e^{-i\delta_k \sum_{n=1}^N \sigma^x_n}e^{-i\gamma_k \sum_{n=1}^N \sigma^z_{n}\sigma^z_{n+1}}\ket{+}$ with depth $p$~\cite{farhi2014quantum}, where we have optimized the $\gamma$ and $\delta$ parameters beforehand.
, a product state $\ket{+}^{\otimes N}$ and a discretized quantum annealing state
\begin{equation}
\ket{\psi_\text{QA}}=\prod_{k=1}^p e^{-iT \sum_{n=1}^NH_x}e^{-i \frac{Tk}{p} H_z}\ket{+}^{\otimes N}\,,
%\ket{\psi_\text{QA}}=\prod_{k=1}^p e^{-iT \sum_{n=1}^Nh\sigma^x_{n}}e^{-i \frac{Tk}{p} \sum_{n=1}^N(\sigma^z_{n}\sigma^z_{n+1}+g\sigma^z_{n})}\ket{+}^{\otimes N}\,,
\end{equation}
where $p$ is the number of layers of the circuit and $T$ the quantum annealing time. The state is constructed by evolving with \eqref{eq:Ising}, where we stepwise evolve with the non-commuting parts $H_x$ and $H_z$. This state is a discretized form of a quantum annealing protocol, where one starts with the ground state of $H_x$, and then slowly increases $H_z$ until one reaches the target Hamiltonian~\eqref{eq:Ising}. In the limit $p\rightarrow\infty$, the adiabatic theorem guarantees that this state becomes the exact ground state~\cite{kadowaki1998quantum}. A similar type of ansatz with additional variational parameters is used for QAOA and VQE~\cite{farhi2014quantum,wiersema2020exploring}.}

To generate the ansatz space for~\eqref{eq:Ansatz_1}, we use a NISQ-friendly adaption of the Krylov subspace approach~\cite{lanczos1950iteration,saad1992analysis,seki2020quantum,motta2020determining}. In the original Krylov subspace approach, the ground state is approximated by a sum of powers of the Hamiltonian $H^k$  applied on the initial state $\ket{\psi}$ prepared on the quantum computer with appropriate coefficients $\alpha_k$ and truncated up to order $K$
\begin{equation}
\vert\xi\left(\alpha\right)\rangle^{\left(K\right)}=\alpha_{0}\vert\psi\rangle+\alpha_{1}H\vert\psi\rangle+\cdots+\alpha_{K}H^{K}\vert\psi\rangle\,.\label{eq:Krylov_Ansatz}
\end{equation}
However, on NISQ computers it is challenging to measure $H^k$. To simplify this approach, we use the fact that the Ising Hamiltonian is a sum of Pauli strings $P_{i}$. We decompose the power of the Hamiltonian $H^k=(\sum_k P_i)^k=\sum_{i_1,\dots,i_k} c_{i_1\dots i_k} P_{i_1}\dots P_{i_k}$ into a sum of products of Pauli strings. We then take each product of Pauli string $P_{i_1}\dots P_{i_k}\ket{\psi}$ and add each unique term to the set $\mathbb{S}$~\cite{,bharti2020iterative}. Note that a product of Pauli strings is again a Pauli string. We do this for each power of the Hamiltonian up to order $K$. The generated set $\mathbb{S}$ contains the original Krylov subspace.  However, it can be easily measured on NISQ devices as the corresponding overlaps are simple measurements of Pauli strings. We show as example the first order of the ansatz states for \eqref{eq:Ising} with $M=3N$ states 
\begin{align*}
\mathbb{S}^1_\text{Ising}=&\{\sigma_1^z\ket{\psi},\dots,\sigma_N^z\ket{\psi},\sigma_1^z\sigma_2^z\ket{\psi},\dots,\sigma_N^z\sigma_1^z\ket{\psi},\\
&\sigma_1^x\ket{\psi},\dots,\sigma_N^x\ket{\psi}\}\,.
\end{align*}
%The $M$-dimensional ansatz space $\mathbb{S}$ is generated via the $K$-moment expansion~\cite{bharti2020iterative}. The cumulative $K$-moment states are generated by applying Pauli operators on the quantum state $\ket{\psi}$. 
%This yields the set $\mathbb{S}=\{\vert \psi \rangle\} \cup \left\{ P_{i_1}\vert\psi\rangle\right\} _{i_1=1}^{r} \cup \dots \cup \left\{ U_{i_K}\dots P_{i_1}\vert\psi\rangle\right\} _{i_1=1,\dots,i_K=1}^{r}$, where $P_i$ are the $r$ Pauli strings that make up $H_\text{ising}=\sum_{i=1}^r s_i P_i$. We pick the first $M$ elements from $\mathbb{CS}_{K}$ to get the ansatz space $\mathbb{S}$. 
We use a subset of $M$ states from $\mathbb{S}$ to run the NSE and investigate the convergence of our approach, where we select the first $M$ states in ascending order of the order of the Krylov subspace $k$. 

\begin{figure}[htbp]
	\centering	
	\subfigimg[width=0.33\textwidth]{(a)}{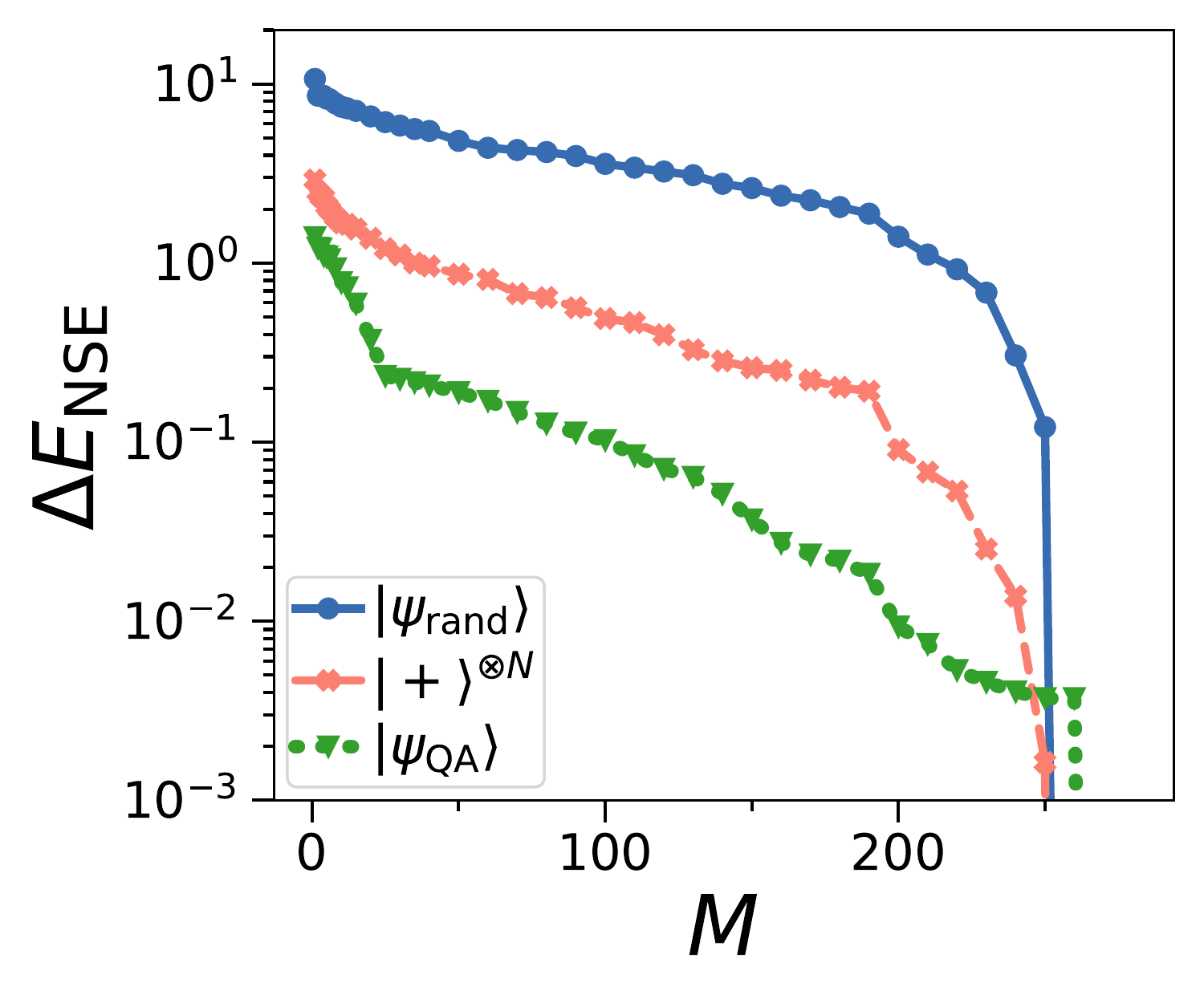}\\
	\subfigimg[width=0.3\textwidth]{(b)}{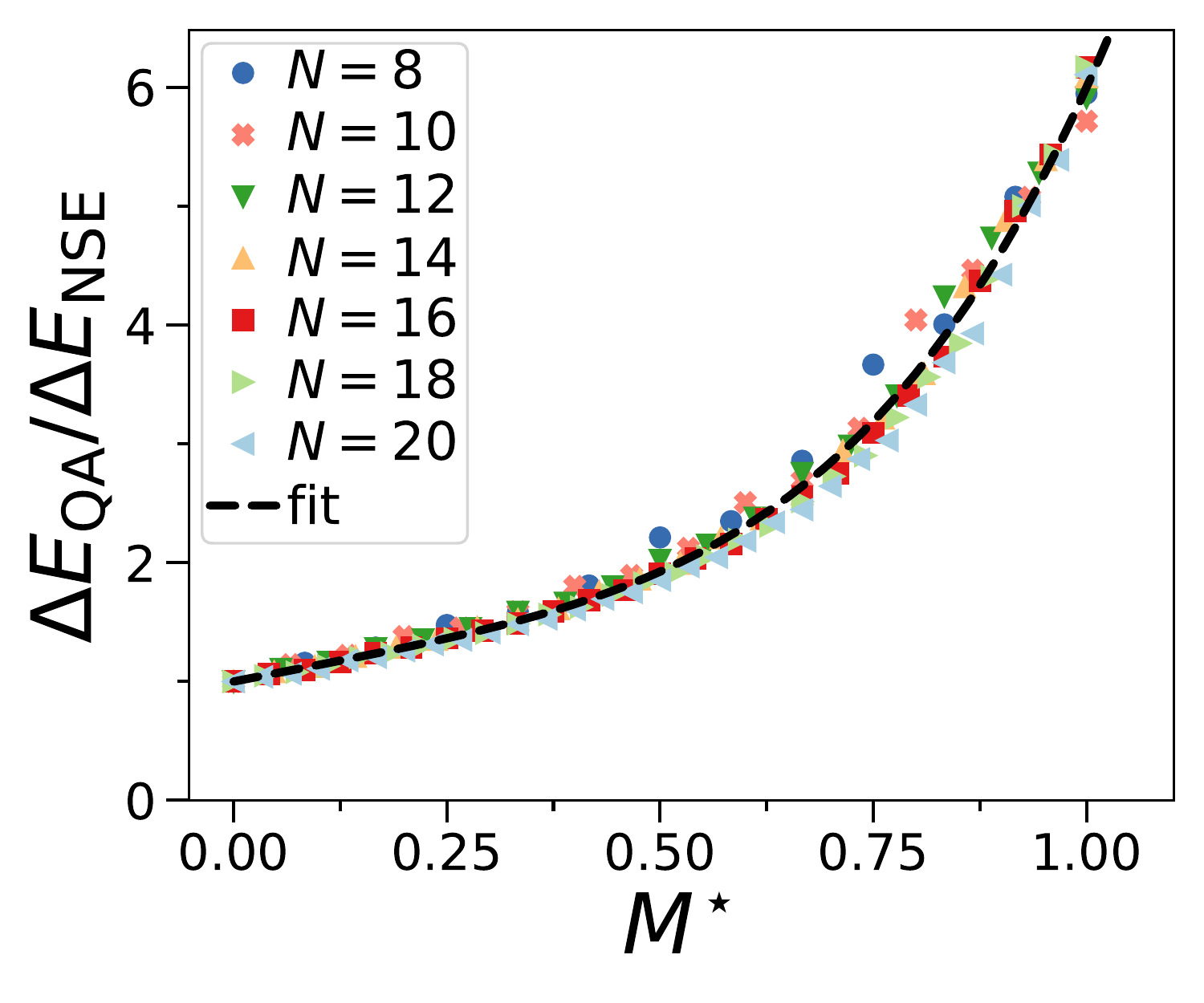}
	\caption{\idg{a} Error $\Delta E_\text{NSE}=E_\text{NSE}-E_\text{g}$ of the energy $E_\text{NSE}$ found via NSE and the exact ground state energy $E_\text{g}$ for the Ising model (\eqref{eq:Ising}, ${h=g=1}$) plotted against number of ansatz states $M$.
	The $M$ ansatz states are generated using a Krylov subspace approach from an initial state. The initial state of $N=8$ qubits is either a hardware efficient circuit $\ket{\psi_\text{rand}}$ composed of $p=4$ layers of random single qubit $y$ rotations and CNOT gates arranged in a chain topology, a product state $\ket{+}^{\otimes N}$ or a quantum annealing state $\ket{\psi_\text{QA}}$ with $p=4$ layers. 
	\idg{b} Scaling with number of qubits $N$ of the relative improvement $\Delta E_\text{QA}/\Delta E_\text{NSE}$ of estimation error of NSE $\Delta E_\text{NSE}$ and quantum annealing $\Delta E_\text{QA}=E_\text{QA}-E_\text{g}$. The relative improvement collapses to a single curve for varying number of ansatz states $M^\star=M/(3N)$ divided by number of qubits. The initial state is the quantum annealing state $\ket{\psi_\text{QA}}$ with $p=N/2$ layers.  We fit with $\Delta E_\text{QA}/\Delta E_\text{NSE}=3.6{M^*}^4+1.2 M^\star+1$. 
	}
	\label{fig:groundstate}
\end{figure}

\revA{In Fig.\ref{fig:groundstate}a, we investigate the error  $\Delta E_\text{NSE}= E_\text{NSE}-E_\text{g}$ between the energy via the NSE $E_\text{NSE}$ and the exact ground state energy $E_\text{g}$ as function of $M$.
We find that with increasing $M$ the energy decreases and beyond a threshold $M$ the exact ground state is reached. The quality of the approximation highly depends on the choice of initial state. We find that the random circuit state $\ket{\psi_\text{rand}}$ converges slowly and only finds the exact ground state when $M$ becomes the same order as the Hilbert space, In contrast, the quantum annealing state converges much faster to a small error and yields a good approximation of the energy even at modest $M$.

In Fig.\ref{fig:groundstate}b, we investigate the scaling of the NSE with number of qubits $N$. We prepare the quantum annealing ansatz $\ket{\psi_\text{QA}}$ with $p=N/2$ layers with energy $E_\text{QA}=\bra{\psi_\text{QA}} H_\text{ising}\ket{\psi_\text{QA}}$. Now, we apply the NSE with initial state $\ket{\psi_\text{QA}}$ to improve the energy estimation by using the first order of the NISQ-friendly Krylov subspace~\eqref{eq:Krylov_Ansatz}.  We show the improvement of NSE in estimating the energy $\Delta E_\text{QA}/\Delta E_\text{NSE}$ relative to the error of quantum annealing $\Delta E_\text{QA}=E_\text{QA}-E_\text{g}$. 
With increasing number of ansatz states $M$, $\Delta E_\text{QA}/\Delta E_\text{NSE}$ improves non-linearly.
For any number of qubits, $\Delta E_\text{QA}/\Delta E_\text{NSE}$ collapses to a single curve as function of number of ansatz states ${M^*}=M/(3N)$ normalised by the number of qubits $N$. This suggests that we can achieve a scalable improvement in energy estimation  $\Delta E_\text{QA}/\Delta E_\text{NSE}(M\propto N)=\text{const}$ when the number of ansatz states $M$ scales linearly with $N$.
For example, we achieve a factor of 6 improvement for $M=3N$ ansatz states for any $N$. This demonstrates that the NSE can substantially enhance the accuracy of finding the ground state energy even for larger system sizes. In the Appendix~\ref{sec:scaling}, we show that the same scaling appears for different parameters of the Hamiltonian and the quantum annealing state. }

\revA{\subsection{Excited states }
We now adapt the NSS solver to find the excited states of the Hamiltonian. We assume that we already found approximations for ground state $\rho_0$ and the first  first $N_{\text{E}}\ge0$ excited states $\rho_n$ with the coefficients $\beta_i$, $i\in\{0,\dots,N_{\text{E}}\}$. To find the $N_{\text{E}}+1$ excited state $\rho_{N_{\text{E}}+1}$ with coefficient $\beta_{N_{\text{E}}+1}$, we run the NSS with the added constraint that the overlap between $\beta_{N_{\text{E}}+1}$ and the already found states is zero with $\sum_{n=0}^{N_{\text{E}}}\text{tr}(\rho_n\rho_{N_{\text{E}}+1})=\sum_{n=0}^{N_{\text{E}}}\text{tr}(\beta_n\mathcal{E}\beta_{N_{\text{E}}+1}\mathcal{E})=0$. The state with the smallest energy that satisfies the constraints is the approximation of the $N_\text{E}+1$ excited state. The NSS solver for excited states is given by
\begin{gather*}
\min_{\beta_{N_{\text{E}}+1}}Tr\left(\beta_{N_{\text{E}}+1}\mathcal{D}\right)\numberthis\label{eq:NSS_excited}\\
\text{s.t. }Tr\left(\beta_{N_{\text{E}}+1}\mathcal{E}\right)=1\\
\sum_{n=0}^{N_{\text{E}}}\text{tr}(\beta_n\mathcal{E}\beta_{N_{\text{E}+1}}\mathcal{E})=0\\
\beta \in \mathcal{H}_+^{M}\,.
\end{gather*}
We can be run Program~\ref{eq:NSS_excited} iteratively with the output of the previous iterations to find up to $M-1$ excited states.}

\revA{\subsection{Symmetry-resolved lowest eigenenergy }
Quantum many-body Hamiltonians $H$ often have symmetries. Each symmetry corresponds to a particular symmetry operator $S$ which commutes with the Hamiltonian $[S,H]=0$. %The symmetry operator and Hamiltonian can be simultaneously diagonalized into block-diagonal form. 
Under time evolution with $H$, the expectation value of $S$ is conserved. A common task encountered in quantum many-body physics and quantum chemistry is to find the eigenstate with lowest energy that is simultaneously an eigenstate of the symmetry $S$ with a conserved quantity $s_k$. 
One common approach to solve this problem is use an ansatz that respects the symmetry operator $S$ initialised with the desired $s_k$~\cite{gard2020efficient,kokail2019self}. However, finding such an ansatz can be difficult and the optimisation is more challenging compared to a general ansatz that breaks the symmetry~\cite{choquette2021quantum,kokail2019self}. 
Another approach is to add penalty terms to the Hamiltonian, but this is known to render the optimization far more challenging~\cite{higgott2019variational,kuroiwa2021penalty}.

Our SDP solver opens up a new way to find the lowest eigenenergy with the conserved quantity. We constrain the minimization problem of the Hamiltonian $H$ to the subspace where $S$ takes the eigenvalue $s_k$ 
\begin{gather*}
\min_\rho Tr\left(\rho H\right)\label{eq:Ham_SDP_primal_conserve}\numberthis\\
\text{s.t. }Tr\left(\rho\right)=1\\
Tr\left(S \rho\right)=s_k\\
Tr\left(S^2 \rho\right)=s_k^2\\
\rho\succcurlyeq0\,.
\end{gather*}
Here, we demand that the expectation value of the symmetry operator and its square are fixed to the conserved quantity, i.e. $\langle S\rangle=s_k$ and $\langle S^2\rangle=s_k^2$. Here, $S^2$ is needed to make sure that we find an eigenstate of the symmetry operator.
We now give the corresponding formulation as NSS
\begin{gather*}
\min_\beta Tr\left(\mathcal{D} H\right)\label{eq:Ham_SDP_primal_conserve_NISQ}\numberthis\\
\text{s.t. }Tr\left(\mathcal{E}\beta\right)=1\\
Tr\left(\mathcal{R} \beta\right)=s_k\\
Tr\left(\mathcal{T} \beta\right)=s_k^2\\
\beta\succcurlyeq 0\,.
\end{gather*}
Here, we define $\mathcal{R}_{a,b}=\bra{\psi_b}S\ket{\psi_a}$ and $\mathcal{T}_{a,b}= \bra{\psi_b}S^2\ket{\psi_a}$, where $S$ and $S^2$ can be decomposed into a sum of unitaries. In contrast to VQE,  SDPs with constraints are convex and can be efficiently solved.

We implement our solver for two important quantum many-body problems, the transverse Ising model \eqref{eq:Ising} and the Heisenberg model
\begin{equation}\label{eq:Heisenberg}
H=\sum_{n=1}^N(\sigma^x_n\sigma^x_{n+1}+\sigma^y_n\sigma^y_{n+1}+h\sigma^z_n\sigma^z_{n+1})\,.
\end{equation}
As demonstration, for the transverse Ising model ($g=0$) we consider its parity symmetry $P=\prod_{n=1}^N\sigma^z_n$ with conserved quantities $p=\pm1$. For the Heisenberg model, we consider the conservation of number of particles $Q=\sum_{n=1}^N\sigma^z_n$ with $q_k=\{-N,-N+2,\dots,N-2,N\}$. We run the NSS with a random hardware efficient circuit and the NISQ-friendly adaption of the Krylov subspace approach to construct the ansatz states.
In Fig.\ref{fig:conserved} we solve for the lowest eigenenergy within a particular symmetry sector for varying number of ansatz states $M$. For low $M$, the constrained NSE does not find any solution as the ansatz is unable to satisify the constraints for the symmetry. Above a specific $M$, we find an appropriate solution, which for further increase of $M$ converges to the lowest eigenenergy $E_0^S$ for the given conserved quantity. For the Heisenberg  model, we find that higher particles $Q$ require a larger number of ansatz states $M$ to find a feasible solution.
We believe that by replacing the random hardware efficient circuit with a better suited initial state, the convergence of the energy can be tremendously improved.}

\begin{figure}[htbp]
	\centering	
	\subfigimg[width=0.33\textwidth]{(a)}{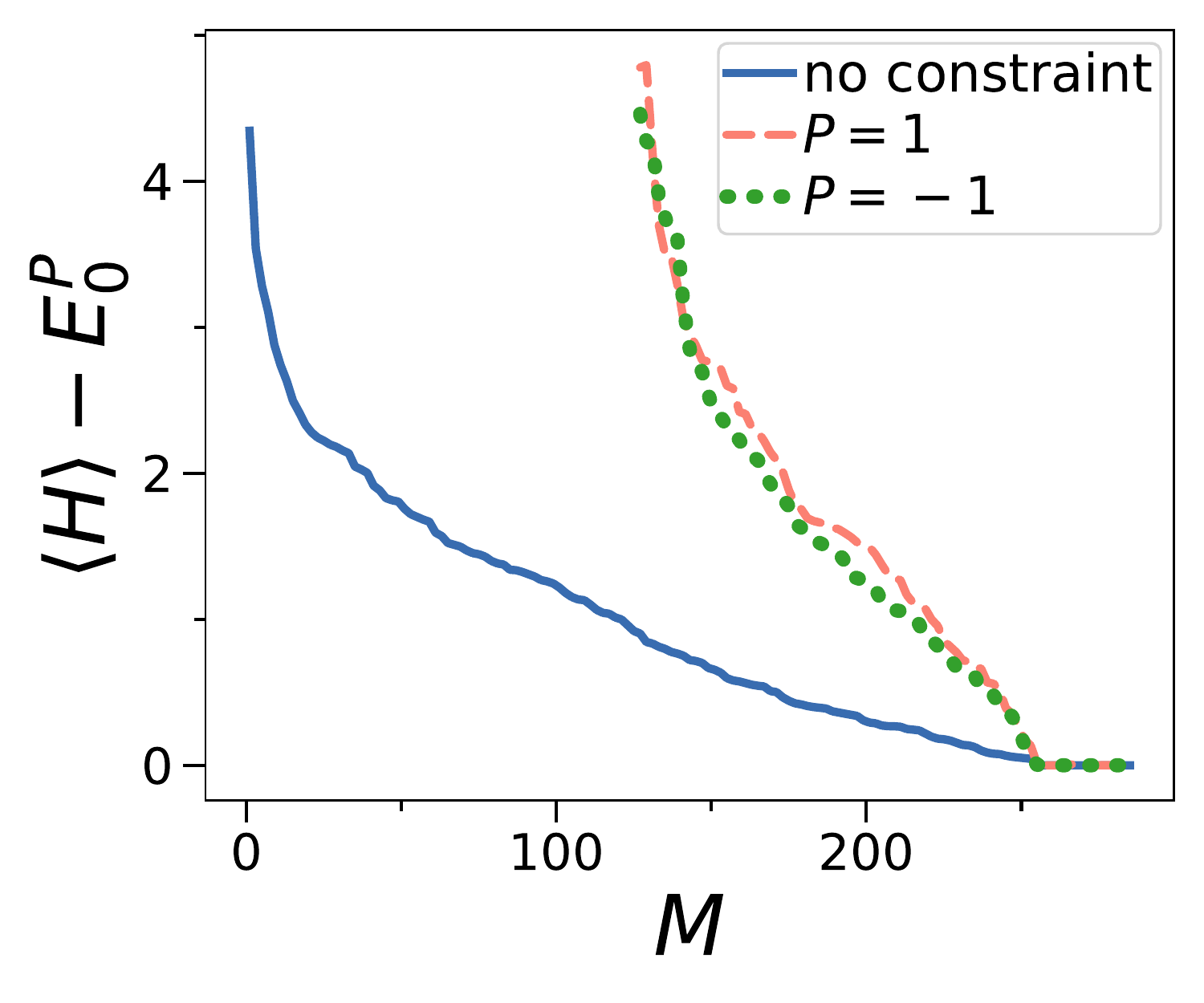}\\
	\subfigimg[width=0.33\textwidth]{(b)}{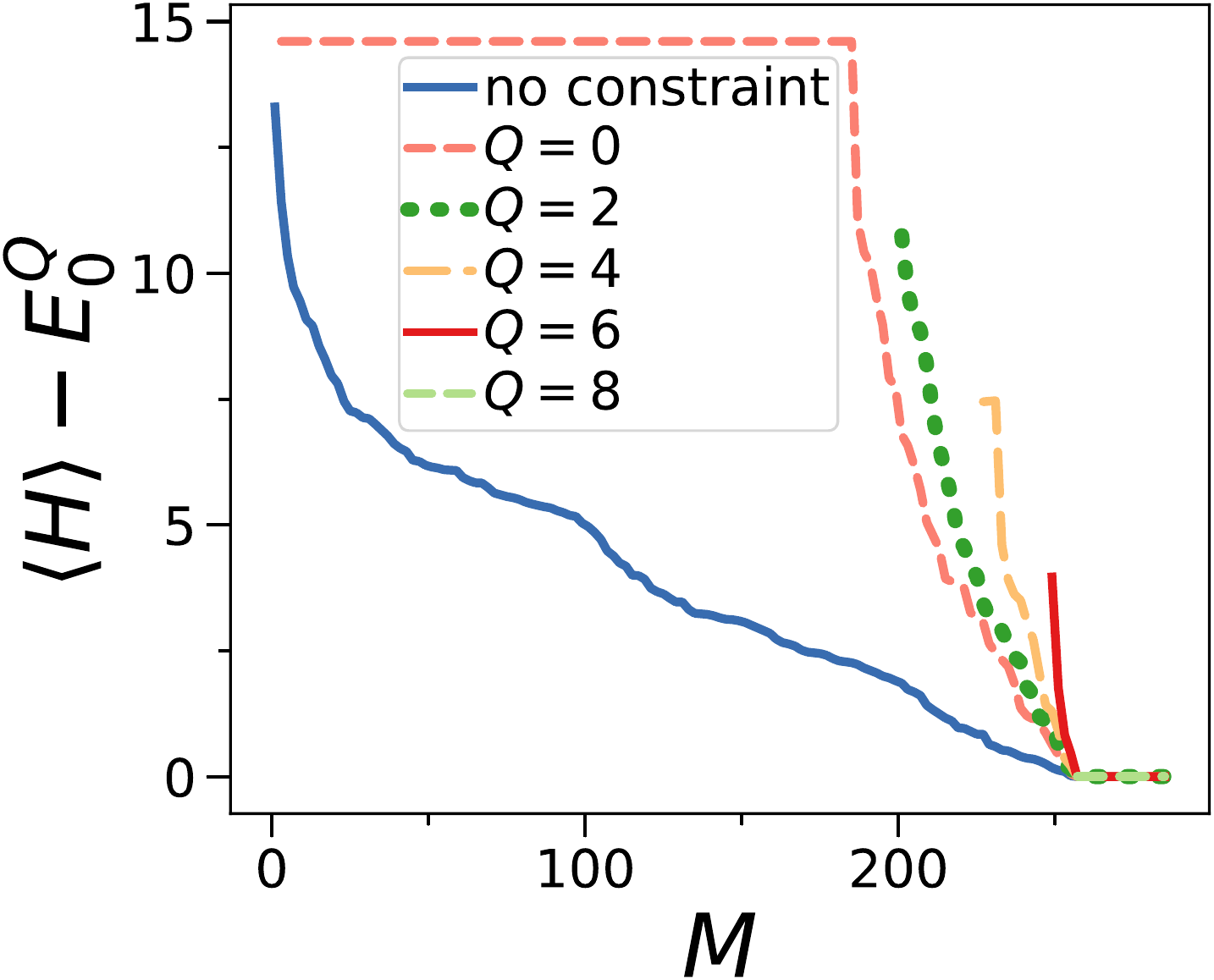}
    \caption{Lowest eigenenergy and eigenstate of a Hamiltonian within a symmetry sector. We solve for the lowest eigenenergy using \eqref{eq:Ham_SDP_primal_conserve_NISQ} where we fix the conserved quantities to a particular value. As reference, we also show the unconstrained NSS solver. We show the energy of the Hamiltonian $\langle H\rangle$ in reference to the lowest eigenergy of the symmetry sector $E_0^S$ as function of the number of ansatz states $M$. The ansatz is randomized quantum circuit.
	\idg{a} One-dimensional transverse Ising Hamiltonian \eqref{eq:Ising} ($h=1$, $g=0$) with parity $P$ and $N=8$ qubits.
	\idg{b} One-dimensional Heisenberg model \eqref{eq:Heisenberg} with conserved number of particles $Q$ for $N=8$ qubits.
	}
	\label{fig:conserved}
\end{figure}

\subsection{Largest Eigenvalue}
Our next task is to find the largest eigenvalue of a sparse matrix $C$ by maximizing Program~\ref{eq:Ham_SDP_Ansatz_Primal_1}.
We assume that $C$ is of size $\mathcal{N}=2^N$ and is represented by a combination of $S$ Pauli string $C=\sum_{i=1}^S c_i P_i^\text{r}$, where the Pauli strings are given by $P^\text{r}_i=\otimes_{j=1}^N \boldsymbol{\sigma}_j$ with $\boldsymbol{\sigma}_j\in\{I,\sigma^x,\sigma^y,\sigma^z\}$ and $c_i$ is a prefactor.  
To numerically demonstrate the performance of NSS, we uniformly sample the Pauli operators for each qubit, choose random $c_i\in[-1,1]$ and use the $N$-bit state with all zeros $\ket{0}^{\otimes N}$ as ansatz state.
Using the Krylov subspace idea for finding the ground state, we similarly construct the ansatz space $\mathbb{S}$.
%Using the idea of cumulative $K$-moment states, we construct the ansatz space $\mathbb{S}$ using products of $P_i^\text{r}$. 
While the expectation values can be calculated classically for the product state, it becomes an intractable problem when using highly entangled quantum states as ansatz states. 
In Fig.\ref{fig:eigval}, we plot the difference between the largest eigenvalue found by NSE and the exact solution $\Delta \lambda$ as function of the number of states $M$ within the ansatz space for different matrix dimensions $\mathcal{N}$. \revA{For small $M$, we find that the error decreases approximately with $\langle \Delta \lambda \rangle\propto M^{-0.77}$.} Beyond a threshold $M$, we find very good convergence even for matrix dimensions $2^{1000}$.

\begin{figure}[htbp]
	\centering	
	\subfigimg[width=0.37\textwidth]{}{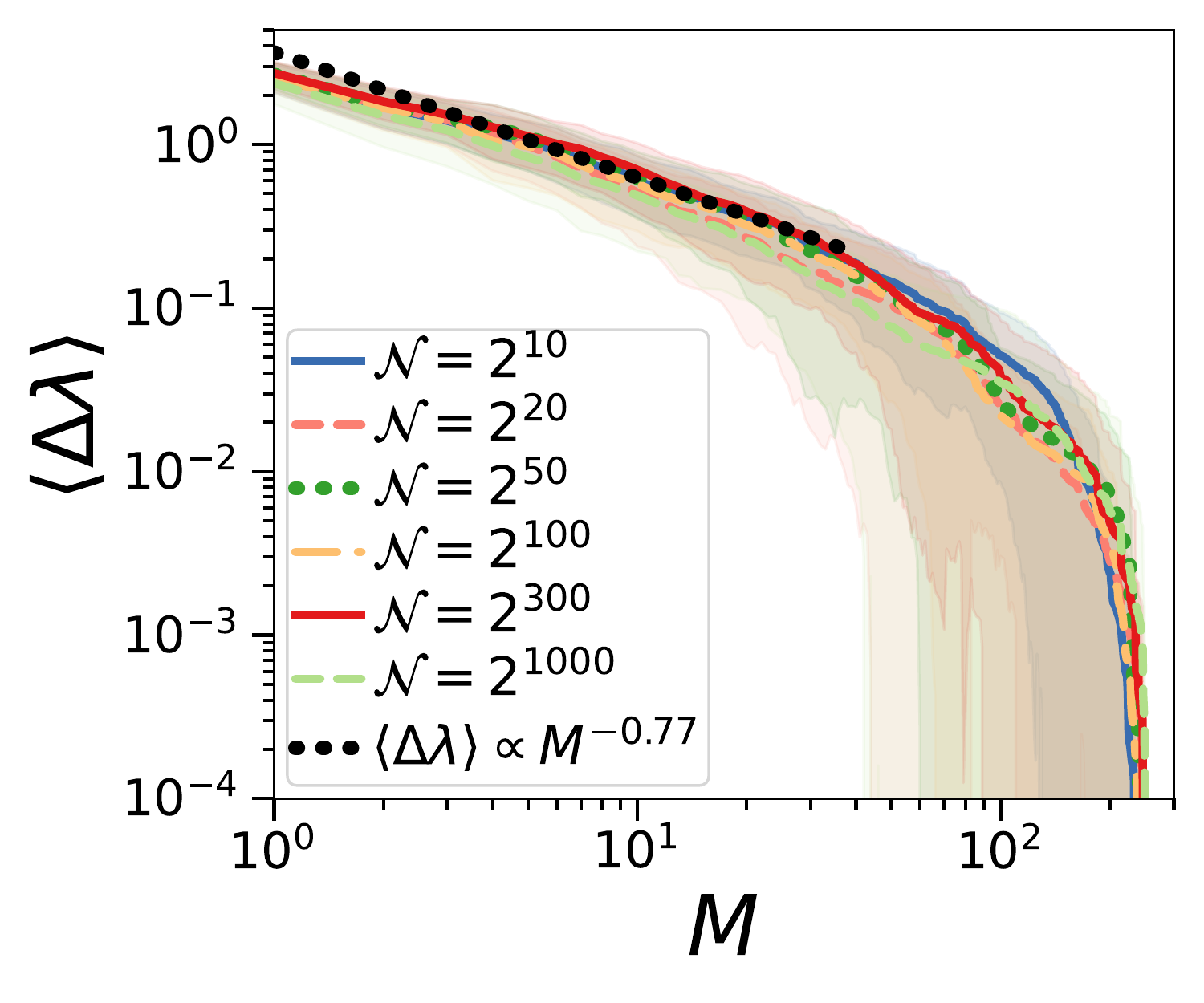}
	\caption{NSS for calculating the largest eigenvalue of a matrix $C$ consisting of $S=8$ random Pauli strings. We show average difference $\langle \Delta \lambda \rangle$ between largest eigenvalue found by NSE and exact largest eigenvalue as function of the size of the ansatz space $M$ for different matrix sizes $\mathcal{N}$. Shaded area is standard deviation of $\Delta \lambda$ averaged over 20 random instances of $C$. The black dots are a fit for small $M$.
	}
	\label{fig:eigval}
\end{figure}

\revA{\subsection{Unambiguous state discrimination}
The task of state discrimination is to identify $N_\text{S}$ states drawn randomly from a set $\mathcal{G}=\{\rho_n\}_{n=1}^{N_\text{S}}$ by performing a measurement on the state.
General measurements are described by positive operator-valued measures (POVMs) $F_n\succcurlyeq 0$ with the condition $\sum_{n}F_n=I$. 
If all the states are pure and pairwise orthogonal, one can simply measure with the projectors on the individual states $\Pi_n=\ket{\psi_n}\bra{\psi_n}$. However, for non-orthogonal states this naive approach will result in classification errors.  
However, we can find a set of $N_\text{S}+1$ POVMs that can unambiguously classify pure states. The POVMs are of the form 
$\sum_{n=1}^{N_\text{S}+1}F_n=I$ with $F_{N_\text{S}+1}=I-\sum_{n=1}^{N_\text{S}}F_n$. 
When we measure the outcome $n\in\{1,\dots,N_\text{S}\}$ associated to POVM $F_n$ with probability $Q_n(\rho)=\text{tr}(\rho F_n)$, we classify the measured state as $\rho_n$. If we measure the outcome $N_\text{S}+1$ with POVM $F_{N_\text{S}+1}$, then we say that we are unable to classify the state.
The problem of finding the optimal set of POVMs can be formulated as an SDP. Our goal is to optimize POVMs in respect to the average probability $Q_\text{correct}=\frac{1}{N_\text{S}}\sum_{n=1}^{N_\text{S}}Q_n(\rho_n)$ to correctly classify the states. Further, we demand that the probability of wrongly classifying state $\rho_k$ is bounded by $Q^\text{error}_k=\sum_{n\ne k}^{N_\text{S}}\text{tr}(\rho_k F_n)\le\epsilon$, $\forall k$ with $\epsilon\ge0$.
The SDP is given by
\begin{gather*}
\max_{F_1,\dots,F_{N_\text{S}}}\frac{1}{N_\text{S}}\sum_{n=1}^{N_\text{S}}Tr\left(F_n\rho_n\right)\numberthis\label{eq:SDP_POVM}\\
\text{s.t. }\sum_{n\ne k}^{N_\text{S}}Tr\left(\rho_k F_n\right)\le \epsilon\quad\forall k\in\{1,\dots,N_\text{S}\}\\
\sum_{n=1}^{N_\text{S}+1} F_n=I\\
F_k \succcurlyeq 0\quad\forall k\in\{1,\dots,N_\text{S}+1\}\,.
\end{gather*}
We can formulate the problem as a NSS. We assume we are given a set of $N_\text{S}$ states to be discriminated written in the form of hybrid states $\rho_n=\sum_{i,j}\beta_{i,j}^n\ket{\psi_i}\bra{\psi_j}$ with the $M\times M$ matrix $\beta^n\succcurlyeq 0$ and $M$ ansatz states $\mathbb{S}=\{\ket{\psi_j}\}_{j=1}^M$. To get valid density matrices, we have $\text{tr}(\rho_n)=\text{tr}(\mathcal{E}\beta^n)=1$, $\mathcal{E}_{a,b}=\braket{\psi_b}{\psi_a}$ and $\beta^n\succcurlyeq 0$. 
We write the $N_\text{S}$ POVMs to be optimized as hybrid POVMs 
\begin{equation}
F_n=\sum_{i,j}\gamma_{i,j}^n\ket{\psi_i}\bra{\psi_j}
\end{equation} 
with $\gamma^n\succcurlyeq 0$. The last POVM for the case where we are unable to classify the state is given by $F_{N_\text{S}+1}=I-\sum_{n=1}^{N_\text{S}}F_n$. We demand that $F_{N_\text{S}+1}\succcurlyeq 0$ is positive semidefinite, which is always fulfilled when
$B=\mathcal{E}-\sum_{n=1}^{N_\text{S}}\mathcal{E}\gamma^n \mathcal{E}\succcurlyeq 0$. 
$F_{N_\text{S}+1}$ is positive semidefinite when for any state $\ket{x}$ we have $\bra{x}F_{N_\text{S}+1}\ket{x}\ge0$. We can write arbitrary states as $\ket{x}=\sum_{n=1}^M\alpha_n\ket{\psi_n}+\alpha_\perp\ket{\psi_\perp}$ with state $\ket{\psi_\perp}$ being orthogonal to the subspace spanned by the ansatz states and normalised coefficients $\boldsymbol{\alpha}$, $\alpha_\perp$. A straightforward calculation shows
$\bra{x}F_{N_\text{S}+1}\ket{x}=\boldsymbol{\alpha}^\dagger B\boldsymbol{\alpha}+\vert\alpha_{\perp}\vert^2\ge\boldsymbol{\alpha}^\dagger B\boldsymbol{\alpha}$, which is non-negative when $B\succcurlyeq 0$. The program of the NSS for state discrimination is now given by
\begin{gather*}
\max_{\boldsymbol{\gamma}^1,\dots,\boldsymbol{\gamma}^{N_\text{S}}}\frac{1}{N_\text{S}}\sum_{n=1}^{N_\text{S}}Tr\left(\boldsymbol{\gamma}^n\mathcal{E}\boldsymbol{\beta}^n\mathcal{E}\right)\numberthis\label{eq:SDP_NISQ_POVM}\\
\text{s.t. }\sum_{n\ne k}^{N_\text{S}}Tr\left(\mathcal{E}\boldsymbol{\beta}^k\mathcal{E}\boldsymbol{\gamma}^n\right)\le \epsilon\quad\forall k\in\{1,\dots,N_\text{S}\}\\
B=\mathcal{E}-\sum_{n=1}^{N_\text{S}}\mathcal{E}\gamma^n \mathcal{E}\succcurlyeq 0\\
\boldsymbol{\gamma}^k \succcurlyeq 0\quad\forall k\in\{1,\dots,N_\text{S}\}\,.
\end{gather*}
As both the states and POVMs are constructed from the same ansatz space, it is unsurprising that Program \ref{eq:SDP_NISQ_POVM} finds the optimal POVMs for unambiguous state discrimination for any number of qubits.

We demonstrate in Fig.\ref{fig:POVM} our algorithm by classifying two pure states generated by a hardware efficient quantum circuit $\ket{\psi_\text{rand}}$ (see Appendix~\ref{sec:hardware_efficient}), which is intractable to simulate for large number of qubits. The two states $\rho_1$ and $\rho_2$ are prepared as $\ket{\psi_k}=\sum_{i,j}^M\beta^k_{ij}P_j\ket{\psi_\text{rand}}\bra{\psi_\text{rand}}P_i$ where $P_i\in\mathbb{S}$ are a random set of $M$ Pauli strings. The matrix $\mathcal{E}$ can be efficiently measured on NISQ computers as shown in Sec.~\ref{sec:NSS}. In Fig.\ref{fig:POVM}a, we show the probability $Q_\text{correct}$ of correctly identifying the states as function of the angle $\phi=\arccos(\sqrt{\text{tr}(\rho_1\rho_2)})$ between the two states. We find that for demanding zero misclassification error $\epsilon=0$, our NSS finds the analytically known optimal POVMs with $Q_\text{correct}^\text{optimal}=1-\cos(\phi)$~\cite{peres1988differentiate}. In Fig.\ref{fig:POVM}b, we show the average probability $Q_\text{unknown}=\frac{1}{N_\text{S}}\sum_{n=1}^{N_\text{S}}\text{tr}(\rho_n F_{N_\text{S}+1})$ that we cannot make a decision. The change in slope of the classification probability for $\epsilon>0$ seen in Fig.\ref{fig:POVM}a coincides with $Q_\text{unknown}=0$.
}

\begin{figure}[htbp]
	\centering	
	\subfigimg[width=0.24\textwidth]{(a)}{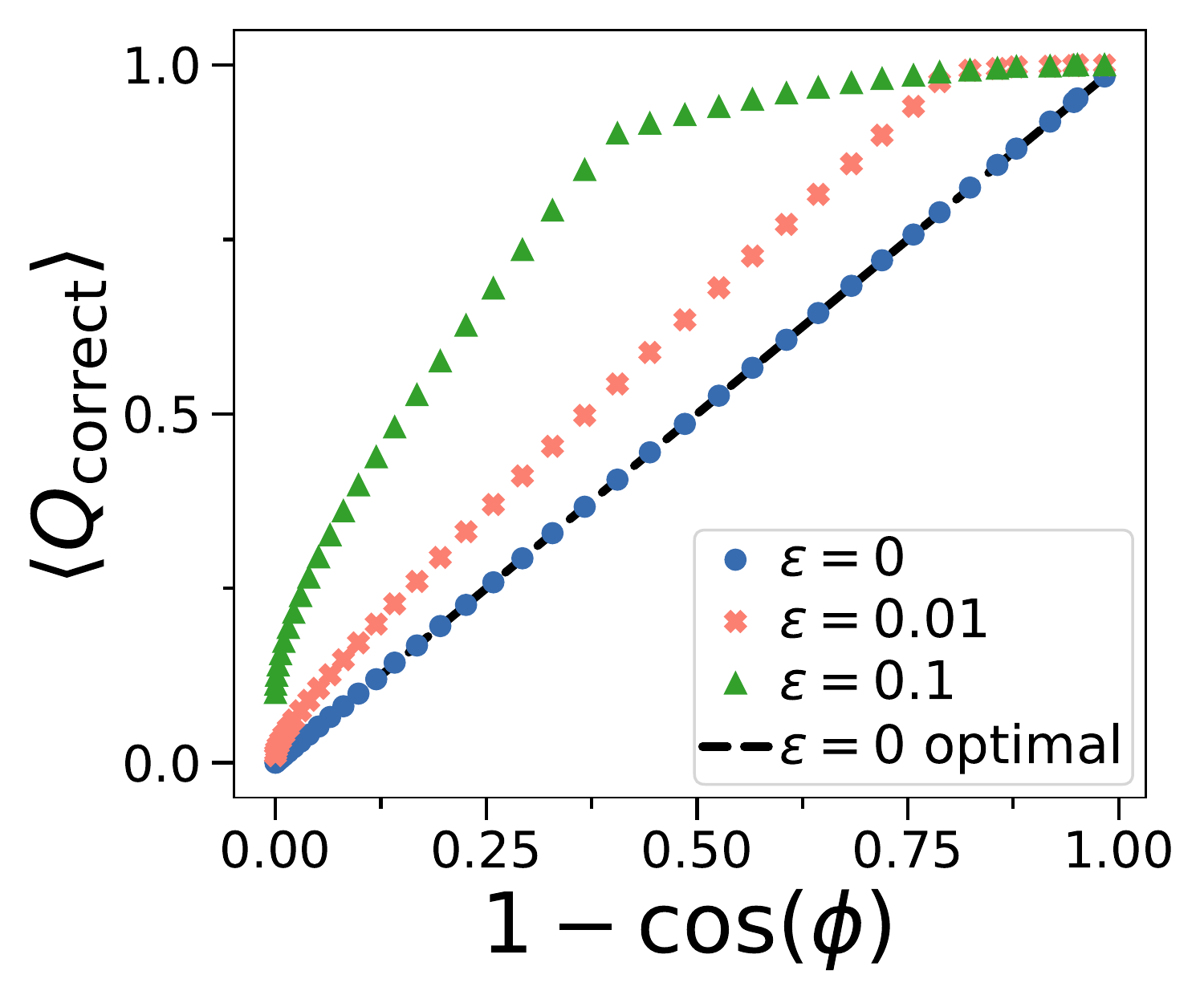}\hfill
	\subfigimg[width=0.24\textwidth]{(b)}{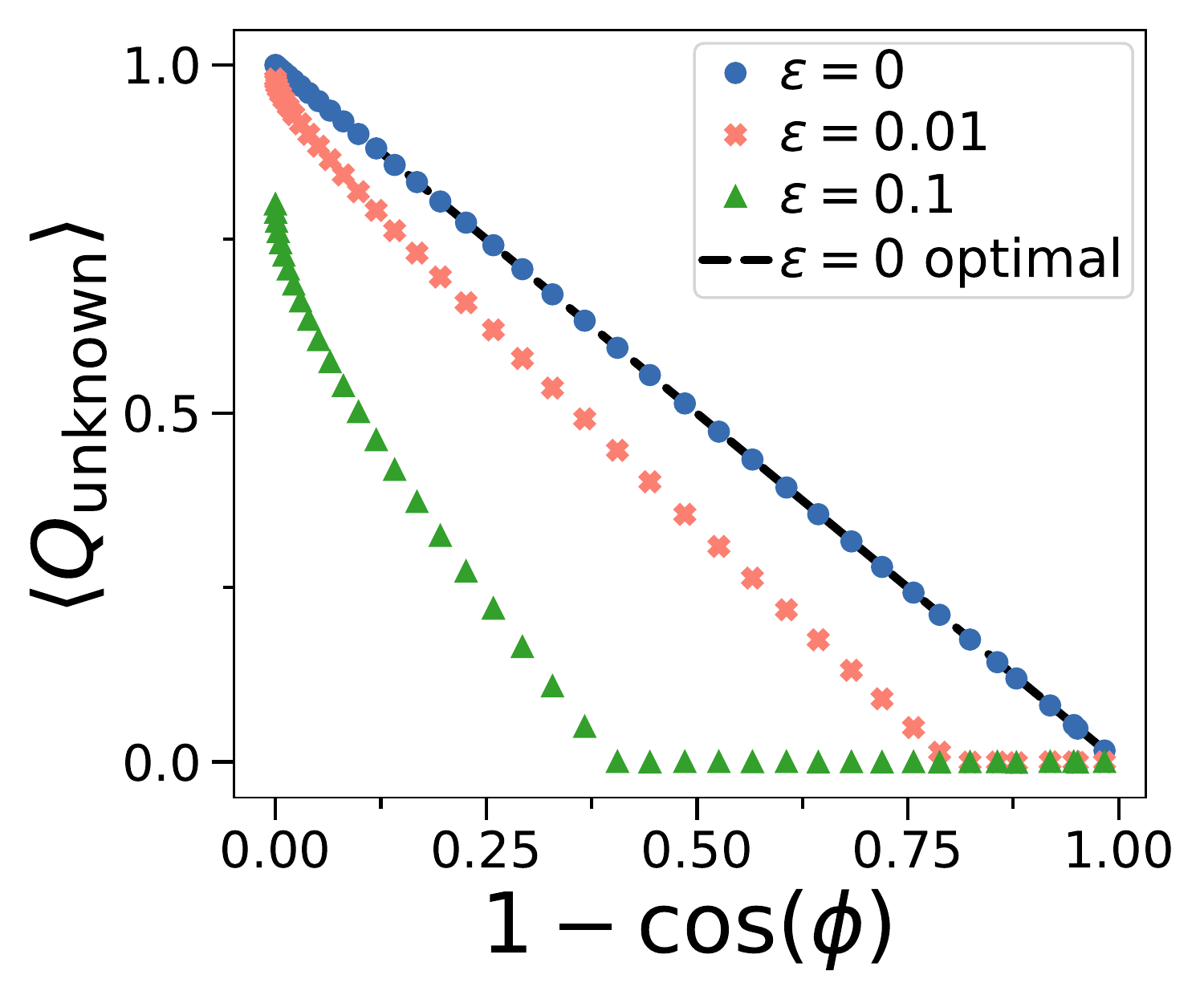}
    \caption{We show the probability of classifying two states with optimized POVMs found with the NSS. We demand that the misclassification error is upper bounded by $\epsilon$. \idg{a} Probability $Q_\text{correct}$ of correctly classifying states  as function of angle between the two states $\phi$. \idg{b} Probability $Q_\text{unknown}$ of being unable to make a decision. The dashed line is the analytic optimal solution for $\epsilon=0$.
    The classified states are hybrid states generated using random quantum circuits with $N=10$ qubits, $d=10$ layers and $M=20$ ansatz states generated with random Pauli strings.
	}
	\label{fig:POVM}
\end{figure}

\subsection{Lovász Theta Number}
Graph invariants are properties that depend only on the  abstract structure of a graph.
The Lovász Theta number is such a graph invariant that was first introduced by László Lovász in the breakthrough $1979$ paper titled ``On the
Shannon capacity of a Graph''~\cite{lovasz1979shannon}. The Lovász Theta number provides
an upper bound to the Shannon capacity of a graph, another graph invariant
quantity. Surprisingly, it is connected with quantum contextuality~\cite{budroni2021quantum,cabello2014graph,bharti2020machine,bharti2019robust,bharti2019local,bharti2021graph} and can help us to understand the potential of quantum computers~\cite{howard2014contextuality}. Given
a graph $G=(V,E)$ with vertex set $V$ and adjacency matrix $E$,
the SDP for the Lovász Theta number is given by
\begin{gather*}
 \max \text{ }Tr\left(JX\right) \numberthis \label{eq:Lovasz_SDP_1}\\
\text{s.t. }X_{i,j}=0\quad\forall E_{i,j}=1\\
Tr(X)=1\\
X\succcurlyeq0\,.
\end{gather*}
Here, $J$ is an all one matrix. Since $X$ is real valued, the ansatz space can be taken as real valued,
which can be achieved within NSS by demanding that the ansatz quantum states are
real valued. 
\begin{figure}
    \centering
    \includegraphics[width=0.28\textwidth]{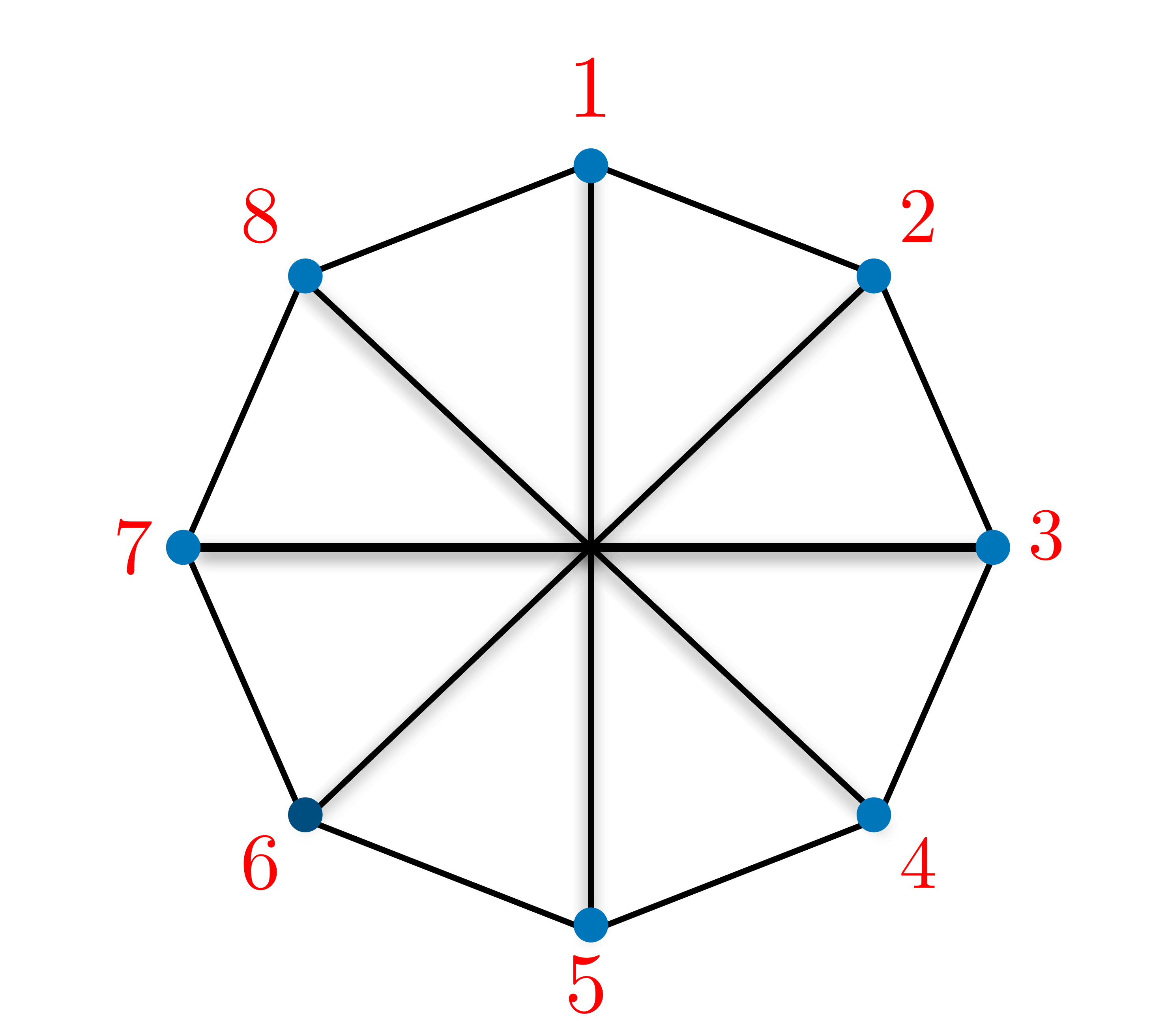}
    \caption{Graph with 8 nodes that has applications in quantum foundations and in device certification protocols. We calculate the Lovász Theta number ($2+\sqrt{2}$) for this graph using NSS.
    }
    \label{fig:circulant}
\end{figure}
To demonstrate the NSS, we calculate the Lovász Theta number for the graph shown in Fig.~\ref{fig:circulant}.
To generate the ansatz states $\mathbb{S}$, we apply a set of Pauli operators on a quantum state. We use the zero state $\ket{0}=\ket{0}^N$ or representative examples of randomized states generated via hardware efficient quantum circuits $\ket{\psi_\text{rand}}$ (see Appendix~\ref{sec:hardware_efficient}). Then, the set of basis states is generated by applying $M$ different combinations of Pauli strings on the state $\mathbb{S}=\left\{P^x_i\ket{\psi}\right\}_{i=1}^M$, where $P^x_i=\otimes_{j=1}^N \boldsymbol{\sigma}_j$ with $\boldsymbol{\sigma}_j\in\{I,\sigma^x\}$.
In Fig.\ref{fig:thetaXOR}a we show the error of the NSS. We calculate the error as the difference between the exact solution $C_\text{exact}$, including the constraints of the problem, and the expectation values $\langle C\rangle$ gained from the quantum state via NSS.
We observe an improvement with increasing number of ansatz states $M$, reaching the optimal solution latest when the number of basis states reaches the dimension of the problem. Depending on the choice of initial state, the optimal solution can be reached with a lower number of ansatz states. 

\begin{figure}[htbp]
	\centering	
	\subfigimg[width=0.24\textwidth]{(a)}{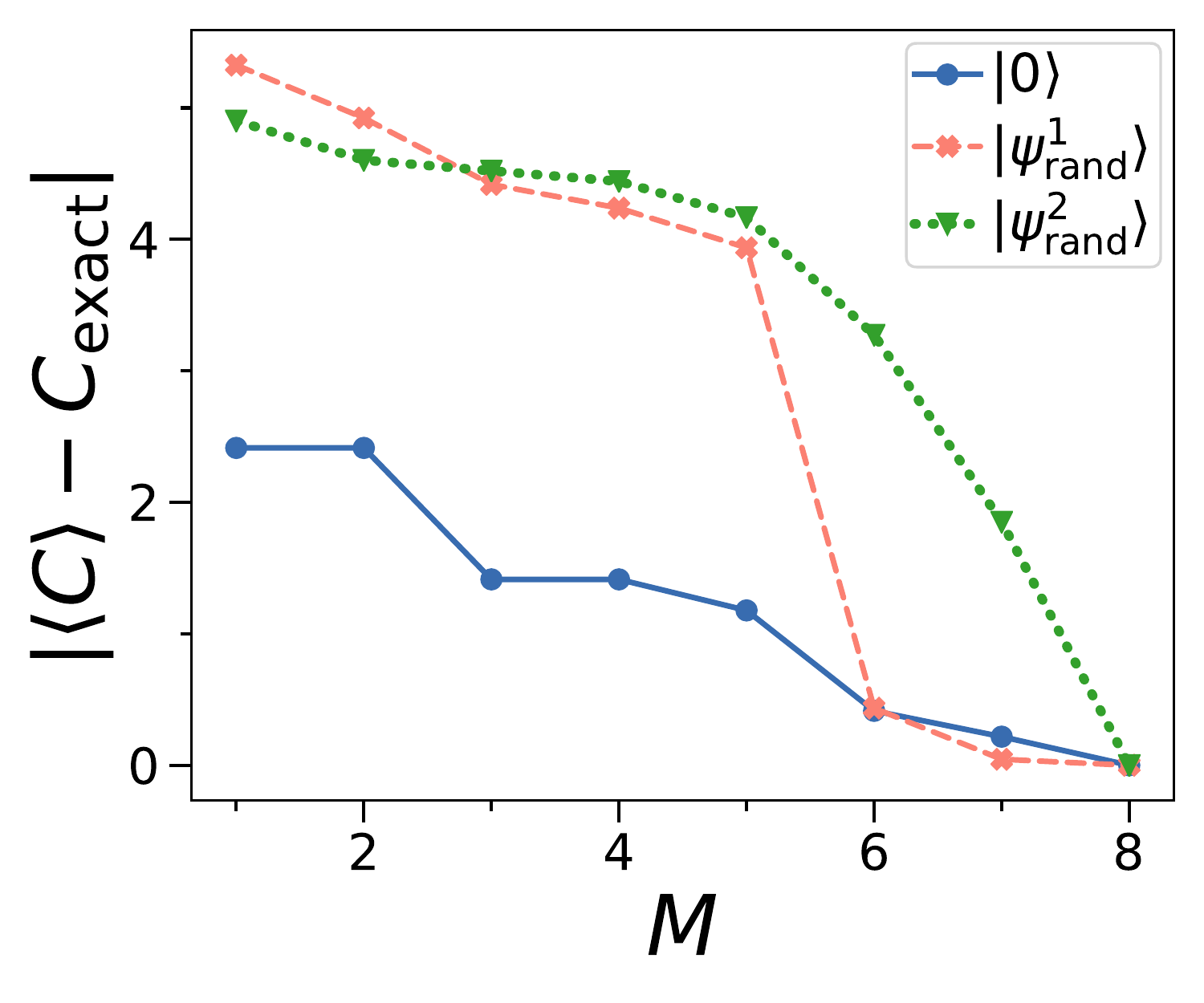}\hfill
	\subfigimg[width=0.24\textwidth]{(b)}{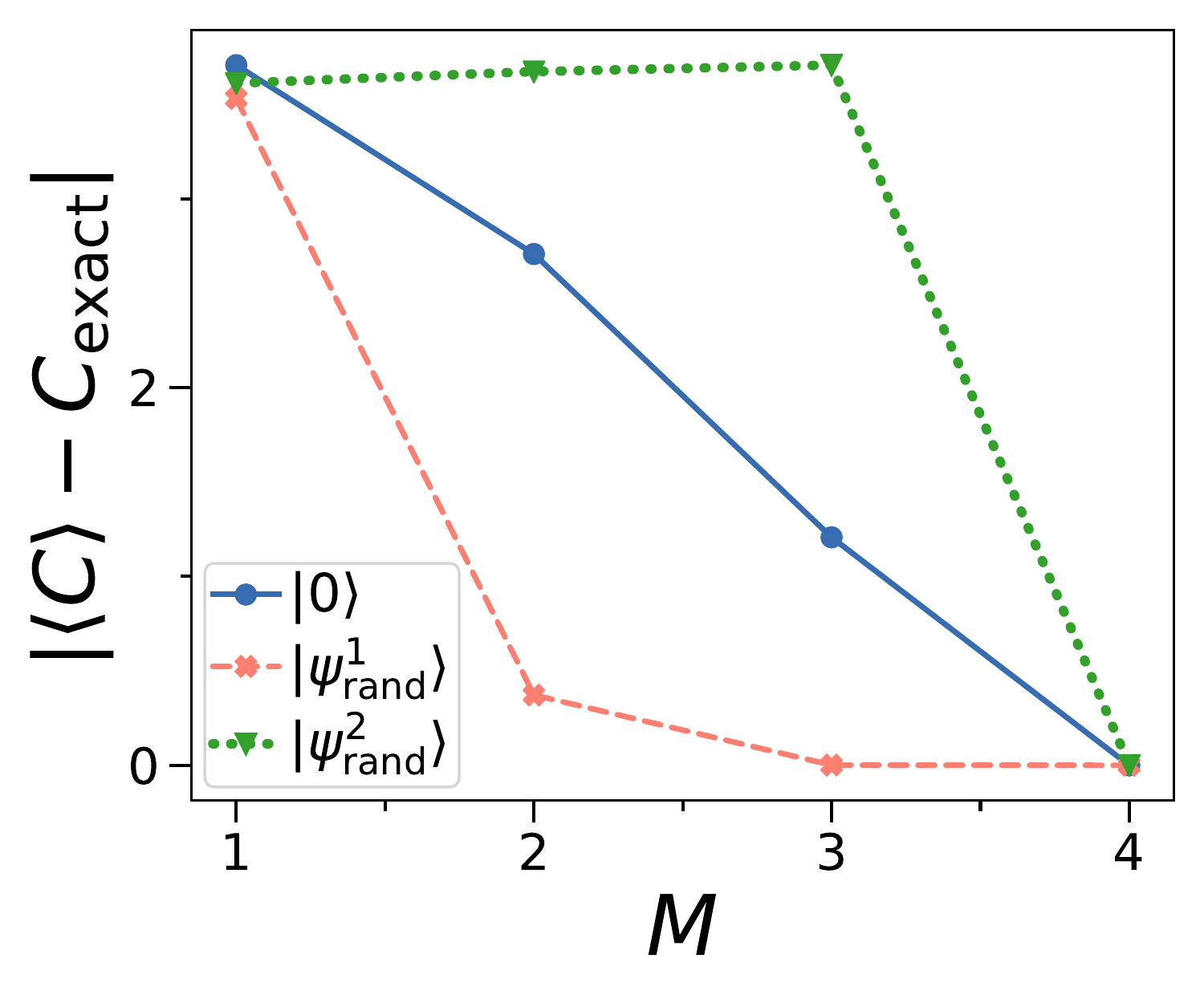}
	\caption{Error of exact solution $C_\text{exact}$ and NSS $\langle C \rangle$ plotted against number of ansatz states $M$. $C_\text{exact}$ and $\langle C \rangle$ are vectors that contain 
    the cost function as well as the constraints to be fulfilled. 
    The ansatz space is generated using the all zero state $\ket{0}$ as well as representative examples of randomized quantum circuit $\ket{\psi_\text{rand}}$ constructed in a hardware efficient manner using $p=4$ layers of single qubit $y$ rotations and CNOT gates arranged in a chain topology. With these states, we generate the $M$ basis states $\mathbb{S}=\left\{P^x_i\ket{\psi}\right\}_{i=1}^M$, where $P^x_i$ is one of the $N$-qubit Pauli strings consisting of identity $I$ and $\sigma^x$ operators.
	\idg{a} NSS algorithm for Lovász Theta number for the graph given in Fig.\ref{fig:circulant} with $N=3$ qubits.
	\idg{b} Bell non-local game with $N=2$ qubits.
	}
	\label{fig:thetaXOR}
\end{figure}

\subsection{Bell Non-Locality}

Finally we apply the NSS to calculate the quantumly achievable
success probability for the canonical Bell non-local game: the Clauser
Horn Shimony Holt (CHSH) game~\cite{Bell64on,clauser69proposed,brunner2014bell,scarani2019bell}. The CHSH game involves two spacelike
separated parties, say Alice and Bob. A referee asks the players uniformly
random pairs of questions $x,y\in\left\{ 0,1\right\} $. The players
have to answer $a,b\in\left\{ 0,1\right\}$ such that
\begin{equation}
a\oplus b=x\land y.\label{eq:CHSH_game_XOR}
\end{equation}
Here $\oplus$ denotes addition modulo $2$ and $\land$ is the logical
AND operator. Using classical strategies, the maximum probability
of success for the CHSH game is upper bounded by $0.75$. However,
using quantum resources such as entangled states, the players can
win the game with probability $\cos^{2}\left(\frac{\pi}{8}\right)$.
The success probability for the CHSH game can be calculated using a
SDP. For details, see Appendix~\ref{sec:xor}. We implemented the  aforementioned
SDP using NSS in Fig.\ref{fig:thetaXOR}b using the same ansatz as for the Lovász Theta number and find that for sufficient number of basis states $M$ we achieve the correct result.

\section{Conclusion}
We presented the NSS for solving SDPs, including rank-constrained SDPs. Our algorithm runs on NISQ devices without the need for classical-quantum feedback loops, requiring only the measurement of overlaps on the quantum computer. 
When the ansatz space is generated by expanding an initial state $\vert \psi \rangle$ with Pauli strings, these overlaps can be measured by sampling the initial state $\vert \psi \rangle$ in a Pauli rotated basis (see Sec.\ref{sec:NSS}).

The ground state problem expressed in terms of density matrices is a SDP and hence a convex optimization program. However, the SDP corresponding to the ground state problem suffers from exponential scaling of the dimension of the quantum state, rendering it difficult to solve on classical computers. To tackle the exponential scaling of the dimension, VQEs employ optimization over smaller dimensional ansatz space. However, the classical optimization corresponding to VQEs is NP-hard, and the landscape contains numerous far from optimal persistent local minima~\cite{bittel2021training,anschuetz2021critical,you2021exponentially}. We employed our NSS to develop the NSE, a NISQ algorithm for the ground state problem.
The key idea of our NSE is to optimize over a smaller number of parameters while preserving the convexity of the original problem. Unlike VQEs, our optimization program is convex, and thus every local minimum is a global minimum. Moreover, the classical optimization program of the NSE is a SDP that can be solved in polynomial-time. \revA{Further, our NSS can efficiently implement constraints on the solution space, which is a well known challenge for VQAs~\cite{kuroiwa2021penalty}.}

The previously proposed quantum assisted eigensolver and iterative
quantum assisted eigensolvers~\cite{bharti2020quantum,bharti2020iterative} are special cases of our rank constrained SDP solver with unit rank (see Appendix~\ref{sec: rank_constrained}).  
With our NSS, it is now possible to solve many important problems that can be formulated in terms of SDPs in a NISQ setting. Our work unlocks the possibility of running one of the most important algorithmic frameworks of classical computing on NISQ computers and exploring the capabilities of the current generation quantum computers. 
\revA{We can use it to find approximation of the ground state and excited states of Hamiltonians, as well as solve symmetry constrained problems. Further, we show how to use NISQ computers to find POVMs that unambiguously discriminate states.}
We also implemented the NSS to calculate Lovász Theta number, a graph invariant with various applications, including in quantum contextuality. Further, we used our algorithm to determine the maximum winning probability of Bell nonlocal games. We demonstrated the applicability of our algorithm for finding the largest eigenvalue of $2^{1000}$ dimensional matrices. \revA{In a recent work~\cite{yu2022quantum}, it was shown that a large family of rank-
constrained SDPs can be written as a convex optimization
over separable two-party quantum states. To solve the aforementioned convex optimization problems, \cite{yu2022quantum} provides a complete hierarchy of SDPs. Based on this result and the techniques in our paper, one can solve the rank-constrained NISQ SDP with a hierarchy of NISQ SDPs.}

Our work leads to many novel avenues for future research. Investigating a systematic problem aware strategy
for determining the initial state $\vert\psi\rangle$
and the set of basis states $\mathbb{S}$ used to construct the
hybrid density matrix will help improve as well as understand the
NSS and its rank constrained variants. It would be fascinating to study our algorithms in the presence of noise~\cite{epperly2021theory}. Further, analysing our algorithms to render complexity-theoretic statements is another exciting direction for further investigation. 

The ground state problem can be thought of as quantum native problem in the sense that the exponential scaling of the Hilbert space size makes it challenging to solve with classical devices. This problem can be framed as a convex optimization program over density matrices. 
Thus, we could conceive the ground state problem as a ``convex quantum native problem". In future, it would be interesting to employ techniques from our work to other convex quantum native problems from disciplines such as quantum chemistry, condensed matter physics and quantum information.
\revA{As we numerically find a clear scaling law in number of qubits for the NSE, we believe rigorous guarantees for the performance of our algorithm for large system sizes can be proven. }

The Quantum interior-point method~\cite{kerenidis2020quantum} and quantum multiplicative weight approaches~\cite{brandao2017quantum,van2017quantum,van2018improvements} have been proposed in the literature to solve SDPs. It would be interesting to develop the corresponding NISQ algorithms. Extending our work to the general case of cone programming seems another exciting direction.

Python code for the numerical calculations performed are available at~\cite{haug2021nisqsdp}.

\medskip
{\noindent {\em Acknowledgements---}} We thank Atul Singh Arora for interesting discussions. We are grateful to the
National Research Foundation and the Ministry of Education, Singapore
for financial support. 
\bibliographystyle{apsrev4-1}
\bibliography{SDP}

\appendix

\section{Hardware Efficient Circuit} \label{sec:hardware_efficient}
In Fig.\ref{fig:AnsatzCircuit}, we show the randomized quantum circuit that we use as one of our ansatz states for our demonstration examples for the NSS. 
\begin{figure}[htbp]
	\centering
	\subfigimg[width=0.26\textwidth]{}{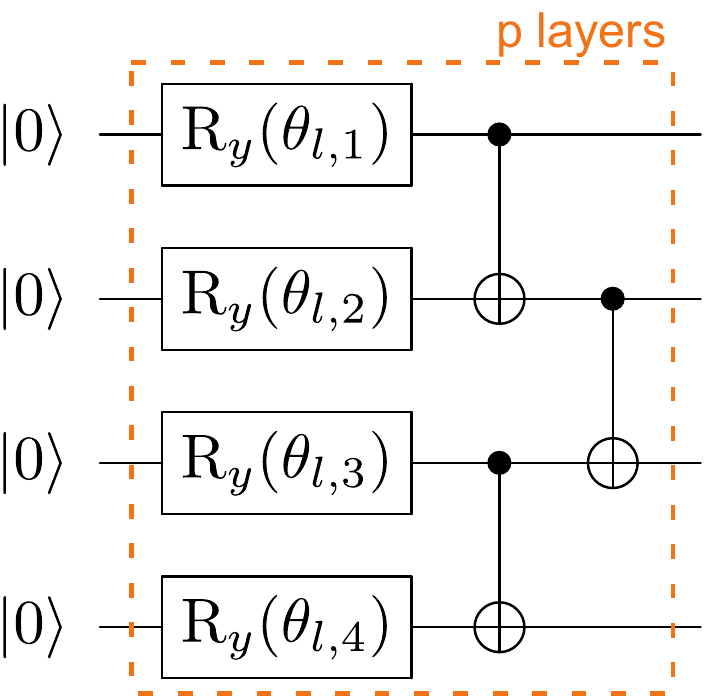}
	\caption{Circuit of $N$ qubits to generate ansatz state $\ket{\psi}$. It consists of $p$ layers of $N$ single qubit rotations around the $y$ axis with randomized parameters $\theta_{l,i}$, followed by CNOT gates arranged in a nearest-neighbor chain topology. }
	\label{fig:AnsatzCircuit}
\end{figure}

\section{Proof for Positive Semidefinte Hybrid Density Matrix} \label{sec:proof_density}
Here, we show that the hybrid density matrix $X_\beta$ is positive semidefinite when the coefficient matrix $\beta$ is positive semidefinite.
\begin{claim}
$X_\beta\succcurlyeq0$ if and only if $\beta\succcurlyeq0.$
\end{claim}
\begin{proof}
We have
\[
X_\beta\succcurlyeq0\iff\sum\beta_{i,j}\langle x\vert\psi_{i}\rangle\langle\psi_{j}\vert x\rangle\geq0\quad\forall\quad\vert x\rangle\in\mathcal{H}.
\]
For a vector $c^{x}$ of size M, defined
via $c_{i}^{x}\equiv\langle x\vert\psi_{i}\rangle$, we have 
\[
X_\beta\succcurlyeq0\iff\left(c^{x}\right)^{\dagger}\beta c^{x}\geq0\quad\forall\quad\vert x\rangle\in\mathcal{H}.
\]
Thus,
\[
X_\beta\succcurlyeq0\iff\beta\succcurlyeq0\,.
\]
\end{proof}

\section{Rank Constrained SDPs} \label{sec: rank_constrained}
We now study rank constrained SDPs. Note that solving rank-constrained
SDPs is NP-hard. The famous max cut problem with $n\times n$ weight matrix $W$ admits
the following rank one constrained SDP.
\begin{gather*}
 \max\text{ }\frac{1}{2}\text{Tr}\left(WX\right)\numberthis \label{eq:rank_1_max_cut} \\
\text{s.t. }X_{i,i}=1\quad\forall i\in[n],\\
X\succcurlyeq0,\\
rank(X) <=1.
\end{gather*}

\subsection{The Rank 1 Case}
We have the following ansatz
\[
x_\alpha=\sum_{i}\alpha_{i}\vert\psi_{i}\rangle
\]
\[
X_\alpha=x_\alpha x_\alpha^{\dagger}
\]
\[
X_\alpha=\sum_{i,j}\alpha_{j}^{\dagger}\alpha_{i}\vert\psi_{i}\rangle\langle\psi_{j}\vert
\]
By construction,
\[
X_\alpha\succcurlyeq0
\]
This corresponds to the positive semidefinite cone constraint and
holds for all values of $\alpha$. Also notice that
\[
X_\alpha^{\dagger}=X_\alpha
\]
by construction and thus $X_\alpha$ is Hermitian. Let us assume that
\[
C=\sum_{k}\beta_{k}U_{k}
\]
and
\[
A_{i}=\sum_{l}f_{i,l}U_{l}^{(i)}.
\]
$Tr\left(CX_\alpha\right)$ translates to
\[
\alpha^{\dagger}\mathcal{D}\alpha,
\]
where
\[
\mathcal{D}_{a,b}=\sum_{k}s_{k}\langle\psi_{b}\vert U_{k}\vert\psi_{a}\rangle.
\]
The constraints $Tr\left(A_{i}X_\alpha\right)= b_{i}$ translate to
\[
\alpha^{\dagger}\mathcal{E}^{(i)}\alpha = b_{i}
\]
where
\[
\mathcal{E}_{a,b}^{(i)}=\sum_{l}f_{i,l}\langle\psi_{b}\vert U_{l}^{(i)}\vert\psi_{a}\rangle.
\]
Thus, in the ansatz space, the standard form rank-constrained SDP reduces to
\begin{equation}
\min\text{ }\alpha^{\dagger}\mathcal{D}\alpha\label{eq:rank_1_SDP}
\end{equation}
\[
\text{s.t. }\alpha^{\dagger}\mathcal{E}^{(i)}\alpha = b_{i}
\]
$\forall i\in\left[m\right].$ This is a quadratically constrained
quadratic program (QCQP). The  quantum assisted eigensolver and iterative
quantum assisted eigensolver~\cite{bharti2020quantum,bharti2020iterative} can be framed as QCQP of the form of program ~\ref{eq:rank_1_SDP}.

\subsection{The Rank k Case}

We have the following ansatz.
\[
x_\alpha^{p}=\sum_{i}\alpha_{i}^{p}\vert\psi_{i}\rangle
\]
\[
X_{\alpha,\gamma} =\sum_{p=1}^{k}\gamma^{p}x_\alpha^{p\dagger}x_\alpha^{p}
\]
\[
X_{\alpha,\gamma}=\sum_{i,j,p}\gamma^{p}\alpha_{j}^{p\dagger}\alpha_{i}^{p}\vert\psi_{i}^{p}\rangle\langle\psi_{j}^{p}\vert
\]
By construction,
\[
X_{\alpha,\gamma}\succcurlyeq0
\]
for $\gamma^{p}\geq0 \quad \forall p \in [k].$ This corresponds to the positive semidefinite
cone constraint and holds for all values of $\alpha^{p}.$ Also notice
that
\[
X_{\alpha,\gamma}^{\dagger}=X_{\alpha,\gamma}
\]
by construction and thus $X_{\alpha,\gamma}$ is Hermitian. Let us assume that
\[
C=\sum_{k}s_{k}U_{k}
\]
and
\[
A_{i}=\sum_{l}f_{i,l}U_{l}^{(i)}.
\]
$Tr\left(C X_{\alpha,\gamma}\right)$ translates to
\[
\sum_{p=1}^{k}\gamma^{p}\alpha^{p\dagger}\mathcal{D}^{p}\alpha^{p},
\]
where
\[
\mathcal{D}_{a,b}^{p}=\sum_{k}s_{k}\langle\psi_{b}^{p}\vert U_{k}\vert\psi_{a}^{p}\rangle.
\]
The constraints $Tr\left(A_{i}X_{\alpha,\gamma}\right) = b_{i}$ translate to
\[
\sum_{p=1}^{k}\gamma^{p}\alpha^{p\dagger}\mathcal{E}^{p(i)}\alpha^{p} = b_{i}
\]
where
\[
\mathcal{E}_{a,b}^{p(i)}=\sum_{l}f_{i,l}\langle\psi_{b}^{p}\vert U_{l}^{(i)}\vert\psi_{a}^{p}\rangle.
\]
Thus, in the ansatz space, the standard form rank-constrained  SDP for general $k$ reduces to
\begin{gather*}
  \min\text{ }\sum_{p=1}^{k}\gamma^{p}\left(\alpha^{p\dagger}\mathcal{D}^{p}\alpha^{p}\right) \numberthis \label{eq:rank_k_SDP} \\
\text{s.t.}\sum_{p=1}^{k}\gamma^{p}\left(\alpha^{p\dagger}\mathcal{E}^{p(i)}\alpha^{p}\right) = b_{i}\quad\forall i\in\left[m\right] \\
\gamma^{p}\geq0 \quad \forall p \in [k].
\end{gather*}
\section{XOR Games SDP Formulation}\label{sec:xor}
\begin{defn}
Two prover game: Given a predicate $V:\mathcal{X}\times\mathcal{Y\times\mathcal{A}}\times\mathcal{B}\rightarrow\{0,1\}$
and a probability distribution $\pi$ on $\mathcal{X}\times\mathcal{Y}$,
a two prover game $\mathcal{G}=\left(\mathcal{X}, \mathcal{Y}, \mathcal{A}, \mathcal{B}, V,\pi\right)$ involves two provers
and one verifier, which proceeds as follows:
\begin{enumerate}
\item The verifier samples a pair of questions $\left(x,y\right)\in\mathcal{X}\times\mathcal{Y}$
according to the probability distribution $\pi$.
\item The verifier sends $x$ and $y$ to the two provers and receives answers
$a\in\mathcal{A}$ and $b\in\mathcal{B}$ respectively.
\item The verifier applies the predicate $V:\mathcal{X}\times\mathcal{Y\times\mathcal{A}}\times\mathcal{B}\rightarrow\{0,1\}$
and accepts the answers if the outcome is $1$, rejects otherwise.
\end{enumerate}
\end{defn}
The size of the sets $\mathcal{A}$ and $\mathcal{B}$, say some integer
value $k$ is assumed to be equal and is referred to as \textit{alphabet
size} of the two prover game.
\begin{defn}
Unique two prover game: A two prover game $\mathcal{G}=\left(\mathcal{X}, \mathcal{Y}, \mathcal{A}, \mathcal{B}, V,\pi\right)$
where the predicate $V:\mathcal{X}\times\mathcal{Y\times\mathcal{A}}\times\mathcal{B}\rightarrow\{0,1\}$
returns value $1$ iff $b=\pi_{x,y}\left(a\right)$, $0$ otherwise
for $\left(x,y\right)\in\mathcal{X}\times\mathcal{Y}$ and outputs
$a,b\in\mathcal{A}\times\mathcal{B}.$ Here $\pi_{x,y}$ is a permutation
of $[k].$
\end{defn}
\begin{defn}
XOR game: A unique two prover game with alphabet size $2$ is known
as XOR game. XOR games are restricted form of two prover nonlocal game where $\mathcal{A}=\mathcal{B}=\left\{ 0,1\right\} $and
the predicate $V$ takes the form
\begin{equation}
V(a,b,x,y)=\begin{cases}
1 & \text{if }a\oplus b=f(x,y)\\
0 & \text{if }a\oplus b\neq f(x,y)
\end{cases}\label{eq:XOR}
\end{equation}
for some given function $f:\mathcal{X\times}\mathcal{Y}\rightarrow\left\{ 0,1\right\} .$
The function $f$ determines whether the two parties should agree
or disgree for each question pair $(x,y)$. 
\end{defn}

The maximum probability of success that the two provers can achieve
is known as value of the game and often denoted by val$\left(\mathcal{G}\right).$ For a given XOR game $\mathcal{G}$ and any strategy, the bias of that strategy is the probability it wins minus probability it loses. The bias of a XOR game  $\mathcal{G}$ is the supremum bias over all possible strategies. Let us denote the supremum bias as $\epsilon(\mathcal{G})$. It is easy to see that
\begin{equation}
    \text{val}\left(\mathcal{G}\right) = 0.5 + 0.5*\epsilon(\mathcal{G})
\end{equation}
For every XOR game $\mathcal{G}$, we further define a matrix $D$ as
\begin{equation}
    D(x,y) = \pi (x,y) (-1)^{f(x,y},
\end{equation}
where $f$ determines the value of the predicate $V$ according to \ref{eq:XOR}
Using $D$, one can further define a symmetric matrix $H$ as
\begin{equation}
H=\frac{1}{2}\begin{pmatrix}0 & D\\
D^{T} & 0
\end{pmatrix}.\label{eq:H_matrix}
\end{equation}
The bias of an XOR game is formulated via $H$ as the following SDP
\begin{gather*}
\max\text{ }Tr\left(HZ\right),\numberthis \label{eq:SDP_XOR}\\
\text{s.t. }Z_{i,i}=1\quad\forall i\in[h],\\
Z\succcurlyeq0,\\
Z\in\mathcal{S}_{+}^{n}.
\end{gather*}
Here, $h$ denotes the size of the $H$ matrix. For the CHSH game, we
have $f_{CHSH}(x,y)=x\land y$ where $\land$ is the logical AND operator.
Thus, the $D$ matrix for the CHSH game is given by
\[
D_{CHSH}=\begin{pmatrix}0.25 & 0.25\\
0.25 & -0.25
\end{pmatrix}.
\]

\section{Scaling of NSE for transverse and longitudinal Ising model}\label{sec:scaling}
\revA{Here, we study the non-integrable Ising model with transverse and longitudinal fields combined with a discretized quantum annealing state as function of number of qubits $N$. We vary the parameters of the Hamiltonian and the number of layers of the quantum annealing state. The initial state has an error $\Delta E_\text{QA}=E_\text{QA}-E_\text{g}$, where $E_\text{QA}$ is the energy of thequantum annealing state and $E_\text{g}$ the exact ground state. We plot the improvement $\Delta E_\text{QA}/\Delta E_\text{NSE}$ of NSE in estimating the ground state. We find for all cases that $\Delta E_\text{QA}/\Delta E_\text{NSE}$ collapses to a single curve for sufficient number of qubits when plotted against number of ansatz states divided by number of qubits ${M^*}=\frac{M}{3N}$. The result is shown in Fig.\ref{fig:groundstatescaling}. We find a collapse to a single curve for all cases for sufficiently large $N$. We find that the NSE yields even better results when increasing the field $h$ of the Ising model or layers $p$ of the quantum annealing state. For $h=\frac{1}{2}$, we find that the collapse becomes evident only for larger number of qubits $N$ compared to the other cases, which we suspect is due to quantum annealing providing worse estimates of the ground state energy when $h$ is small.}

\begin{figure*}[htbp]
	\centering
	\subfigimg[width=0.33\textwidth]{(a)}{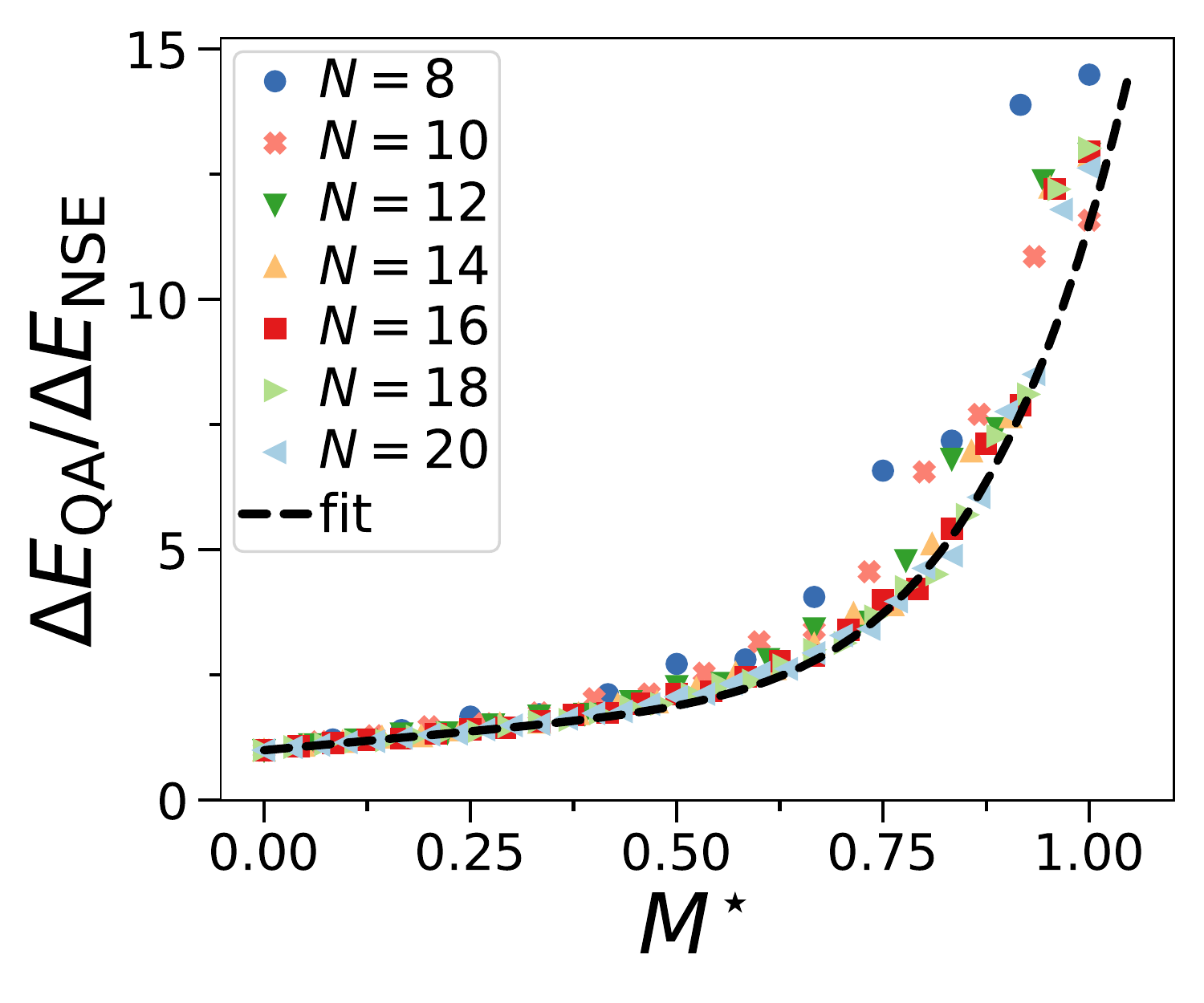}
	\subfigimg[width=0.33\textwidth]{(b)}{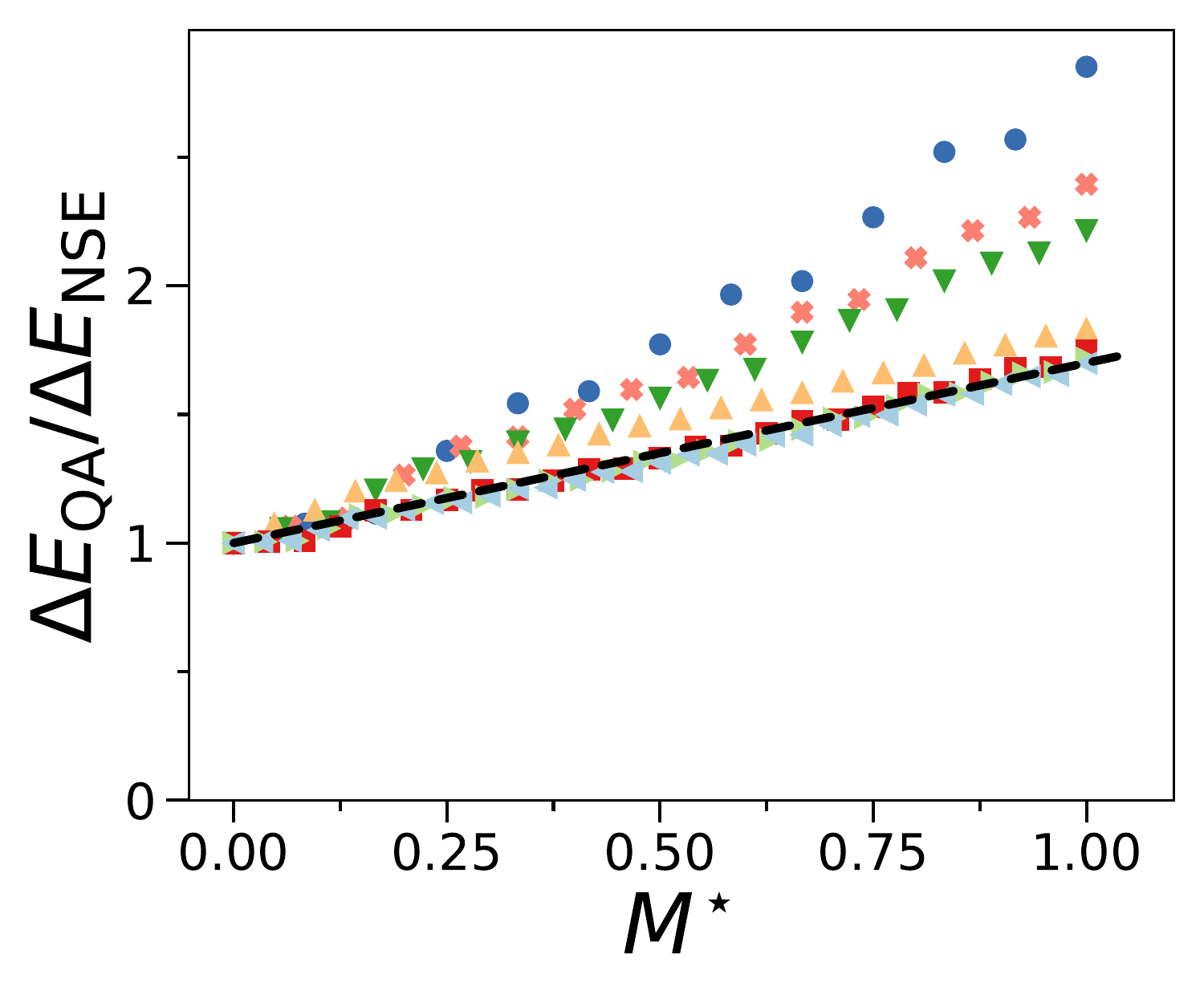}\\
	\subfigimg[width=0.33\textwidth]{(c)}{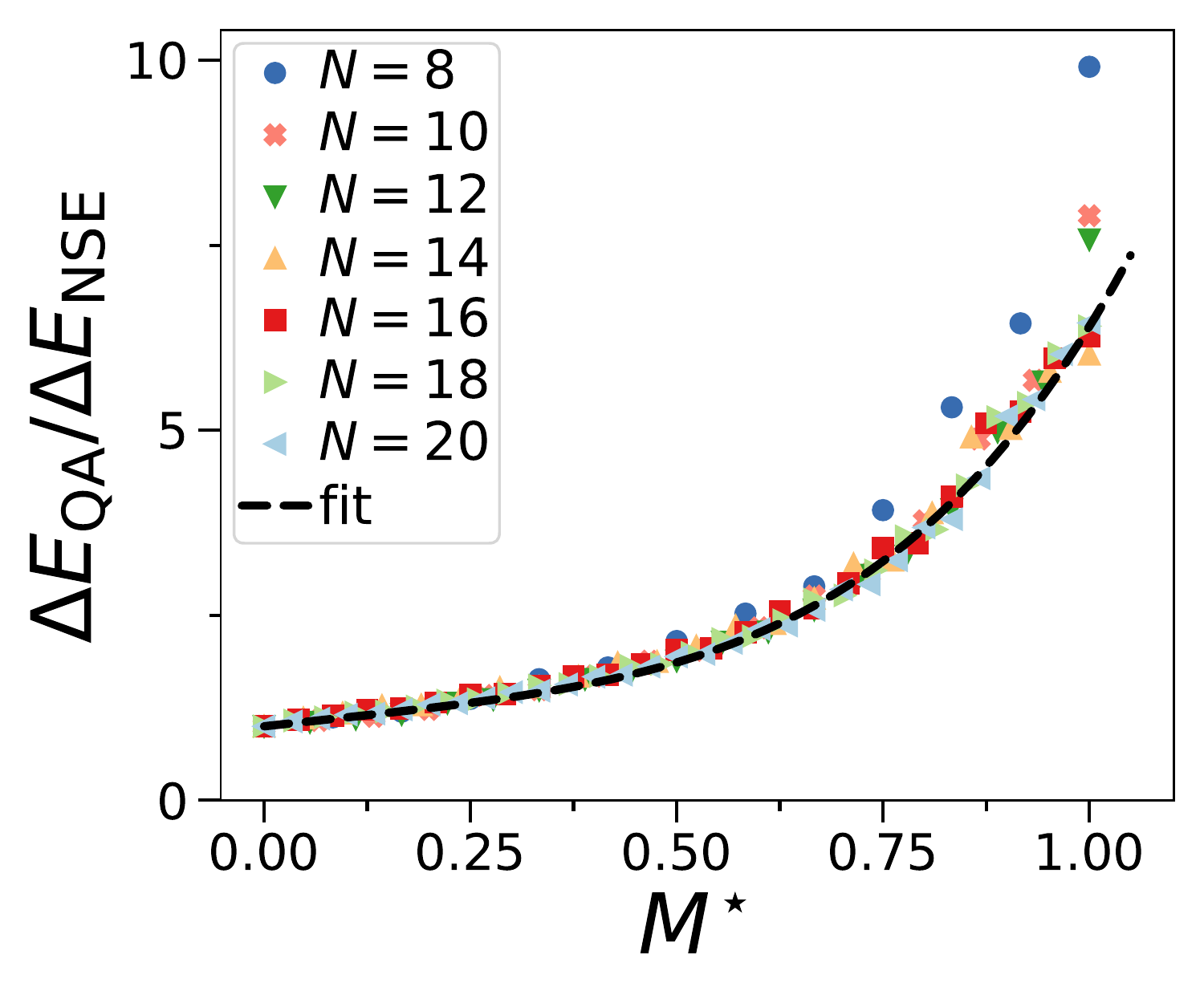}
	\subfigimg[width=0.33\textwidth]{(d)}{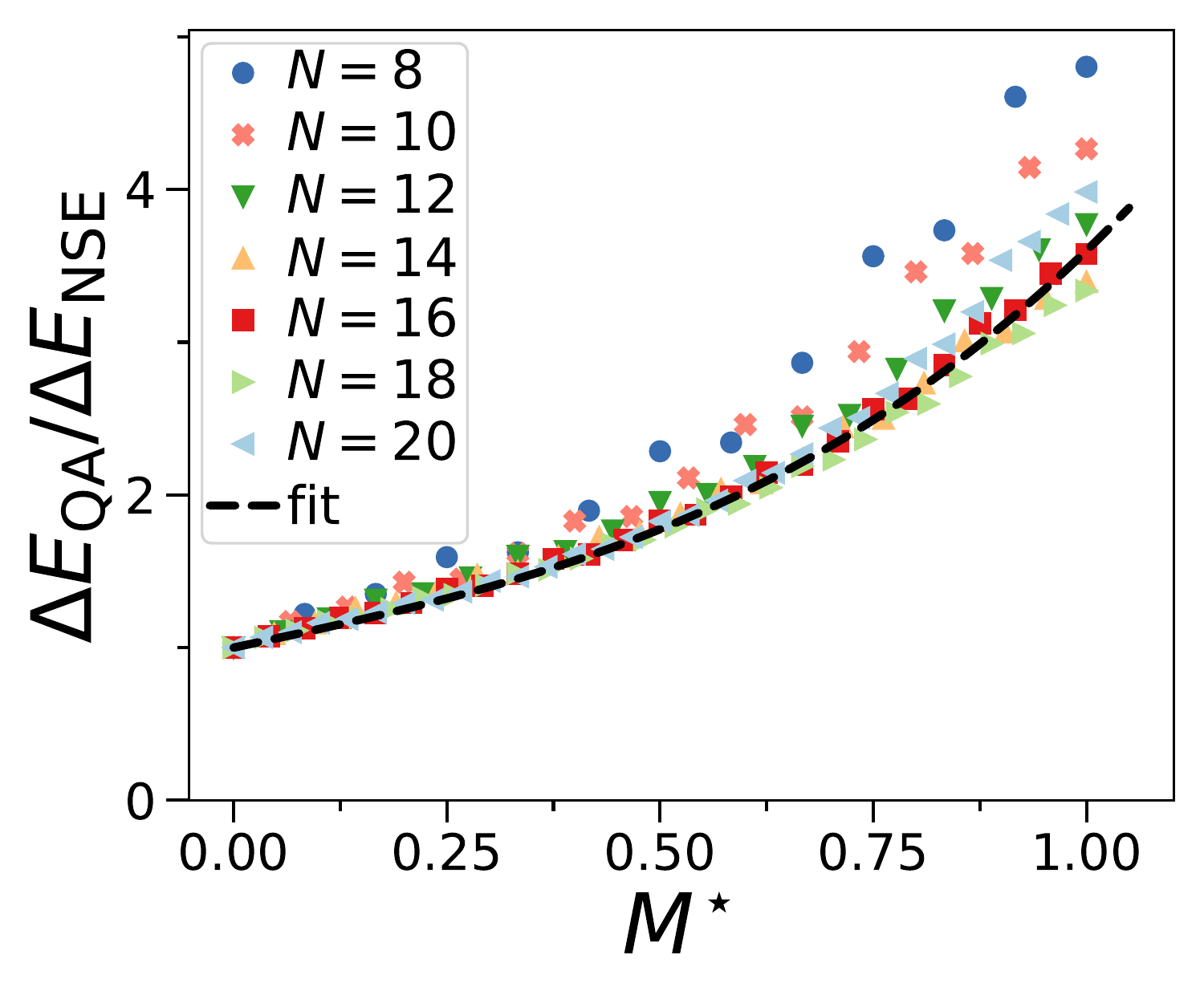}
	\caption{
	We show the scaling with $N$ of the relative improvement $\Delta E_\text{QA}/\Delta E_\text{NSE}$. Here, $\Delta E_\text{NSE}$ is the error $\Delta E_\text{NSE}=E_\text{NSE}-E_\text{g}$ of the NSE and $\Delta E_\text{QA}=E_\text{QA}-E_\text{g}$ the error of the quantum annealing state. The relative improvement collapses to a single curve for varying number of ansatz states ${M^*}=M/(3N)$ divided by number of qubits. We vary transverse field $h$ and  number $p=cN$ of layers of quantum annealing. \idg{a} $h=2$ and $p=N/2$ with fit $\Delta E_\text{QA}/\Delta E_\text{NSE}=9{M^*}^6+1.5{M^*}+1$. \idg{b}  $h=\frac{1}{2}$  and $p=N/2$   with fit $E_\text{QA}/\Delta E_\text{NSE}=0.7{M^*}+1$.
	\idg{c} $h=1$ and $p=N$ with fit $\Delta E_\text{QA}/\Delta E_\text{NSE}=4.2{M^*}^4+1.2{M^*}+1$.
	\idg{d} $h=1$ and $p=N/4$ with fit $\Delta E_\text{QA}/\Delta E_\text{NSE}=1.4{M^*}^3+1.2{M^*}+1$. We fix the field $g=1$ for all curves. }
	\label{fig:groundstatescaling}
\end{figure*}

\end{document}